\documentclass[runningheads]{llncs}
\usepackage[T1]{fontenc}
\usepackage{graphicx}
\usepackage{array}
\usepackage{amssymb}
\usepackage{amsmath}
\usepackage{stmaryrd}
\usepackage{subcaption}
\usepackage{tikz}
\usepackage{enumitem} 
\usepackage{float}
\usepackage{url}
\usepackage{circle}
\usepackage{upgreek}
\usepackage[all]{xy}
\usepackage{listings}
\usepackage{algorithm}
\usepackage{pgfplots}
\pgfplotsset{compat=1.18}
\usepackage[noend]{algpseudocode}
\newcommand{\equivclass}{\lambda}
\newcommand{\equivset}{\Lambda}
\newcommand{\pay}{payof\!f}

\begin{document}
\title{Runtime Verification via Rational Monitor with Imperfect Information}
%
\titlerunning{RV via Rational Monitor with Imperfect Information}
%
\author{Angelo Ferrando\inst{1}\orcidID{0000-0002-8711-4670} \and
Vadim Malvone\inst{2}\orcidID{0000-0001-6138-4229}}
\authorrunning{A. Ferrando and V. Malvone}
%
\institute{Department of Physics, Informatics and Mathematics, University of Modena and Reggio Emilia, Italy \and
LTCI, Telecom Paris, Institut Polytechnique de Paris, France
\email{angelo.ferrando@unimore.it}\\
\email{vadim.malvone@telecom-paris.fr}}
\maketitle              
\newcommand{\ev}{ev}
\newcommand{\univers}{\mathbb{U}}
\newcommand{\natSet}{\mathbb{N}}
\newcommand{\boolSet}{\mathbb{B}}
\newcommand{\presumablyTrue}{?_{\top}}
\newcommand{\presumablyFalse}{?_{\bot}}
\newcommand{\unknown}{?}
\newcommand{\giveup}{\upchi}
\newcommand{\cntxt}{\kappa}
\newcommand{\combine}{\chi}
\newcommand{\refine}{\otimes}
\newcommand{\aggregate}{\circ}
\newcommand{\multimodelSet}{\modelSet^{M}}
\newcommand{\compose}{\diamond}
\newcommand{\setting}{\mathcal{ST}}
\newcommand{\enrich}[1]{\mathsterling({#1})}
\newcommand{\decisionSet}{D}
\newcommand{\emptyTrace}{\epsilon}
\newcommand{\trace}{\sigma}
\newcommand{\continuation}{u}
\newcommand{\inftrace}{u}
\newcommand{\systemTraces}{tr(\eventSet)}
\newcommand{\spec}{\varphi}
\newcommand{\specSet}{\Phi}
\newcommand{\modelSet}{\Psi}
\newcommand{\functionSet}{\mathcal{F}}
\newcommand{\system}{S}
\newcommand{\eventSet}{\Sigma}
\newcommand{\predicate}{P}
\newcommand{\project}{\pi}
\newcommand{\sem}[1]{\llbracket {#1} \rrbracket}
\newcommand{\component}{\mathcal{C}}
\newcommand{\model}[1]{\ifx&#1&{\psi}\else \psi_{#1}\fi}
\newcommand{\function}[1]{\ifx&#1&{F}\else F_{#1}\fi}
\newcommand{\compMonitor}[1]{CoMon_{{#1}}}
\newcommand{\compMonitorAppl}[2]{CoMon_{{#1}}({#2})}
\newcommand{\decompose}{decomp}
\newcommand{\monitor}[3]{Mon_{{#1}, {#2}}^{#3}}
\newcommand{\choreographedMonitor}[3]{ChorMMon_{{#1}, {#2}}^{#3}}
\newcommand{\orchestratedMonitor}[3]{OrchMMon_{{#1}, {#2}}^{#3}}
\newcommand{\cenMonitor}[3]{CenMMon_{{#1}, {#2}}^{#3}}
\newcommand{\monitorAppl}[4]{Mon_{{#1}, {#2}}^{#3}({#4})}
\newcommand{\choreographedMonitorAppl}[4]{ChorMMon_{{#1}, {#2}}^{#3}({#4})}
\newcommand{\orchestratedMonitorAppl}[4]{OrchMMon_{{#1}, {#2}}^{#3}({#4})}
\newcommand{\cenMonitorAppl}[4]{CenMMon_{{#1}, {#2}}^{#3}({#4})}
\newcommand{\stmonitor}[1]{Mon_{{#1}}}
\newcommand{\stmonitorAppl}[2]{Mon_{{#1}}({#2})}
\newcommand{\stmonitorv}[1]{Mon_{{#1}}^v}
\newcommand{\stmonitorApplv}[2]{Mon_{{#1}}^v({#2})}
\newcommand{\modelEvents}[1]{\eventSet_{\model{#1}}}
\newcommand{\modelSetEvents}{\eventSet_\modelSet}
\newtheorem{observation}{\textsc{Observation}}
\newcommand{\modelTraces}[1]{tr(\model{#1})}
\newcommand{\modelSetTraces}{tr(\modelSet)}
\newcommand{\traces}[1]{tr({#1})}
\newcommand{\trans}[1]{\stackrel{{#1}}{\rightarrow}}
\newcommand{\combinedModel}{{\model{c}}}
\newcommand{\lang}[1]{\mathcal{L}({#1})}
\newcommand{\explicit}[1]{\epsilon({#1})}

\newcommand{\always}{\square}

\newcommand{\univmonitorable}{\forall_{PZ}\textrm{-}monitorable}
\newcommand{\existmonitorable}{\exists_{PZ}\textrm{-}monitorable}
\newcommand{\univmonitorableclass}{\forall_{PZ}}
\newcommand{\existmonitorableclass}{\exists_{PZ}}

\newcommand{\monitorable}[1]{{#1}\textrm{-}monitorable}

\newcommand{\signedednote}[2]{
\begin{center} \begin{minipage}{3in}
#2 --    #1
\end{minipage} \end{center}}
\newcommand{\vad}[1]{{\color{blue}\signedednote{Vadim}{#1}}}
\newcommand{\ang}[1]{{\color{blue}\signedednote{Angelo}{#1}}}
\newcommand{\new}[1]{{\color{black}{#1}}}

\newcommand{\until}{\ensuremath{\mathop{\mathrm{\mathbf{U}}}}}
\newcommand{\release}{\ensuremath{\mathop{\mathrm{\mathbf{R}}}}}
\newcommand{\nextOp}[1]{\Circle{#1}}

\begin{abstract}
Trusting software systems, particularly autonomous ones, is challenging. To address this, formal verification techniques can ensure these systems behave as expected. Runtime Verification (RV) is a leading, lightweight method for verifying system behaviour during execution. However, traditional RV assumes perfect information, meaning the monitoring component perceives everything accurately. 
This assumption often fails, especially with autonomous systems operating in real-world environments where sensors might be faulty. Additionally, traditional RV considers the monitor to be passive, lacking the capability to interpret the system's information and thus unable to address incomplete data.
In this work, we extend standard RV of Linear Temporal Logic properties to accommodate scenarios where the monitor has imperfect information and behaves rationally. We outline the necessary engineering steps to update the verification pipeline and demonstrate our implementation in a case study involving robotic systems.
\end{abstract}

\section{Introduction}
\label{sec:introduction}

Developing high-quality software is a demanding task for various reasons, including complexity and the presence of autonomous behaviours~\cite{DBLP:journals/corr/MiguelMR14}. Techniques designed for monolithic systems may not be effective for distributed and autonomous ones, posing both technological and engineering challenges. Over the past decades, we have witnessed significant advances in software engineering research, particularly in software development. However, the need for re-engineering extends beyond development to include software verification. As software evolves, so must the methods used to verify it. Runtime Verification (RV), like other verification techniques, must adapt to these changes.

Runtime Verification (RV)~\cite{DBLP:journals/jlp/LeuckerS09,DBLP:series/lncs/BartocciFFR18} is a formal verification technique used to verify the runtime behaviour of software and hardware systems. Unlike other verification methods, RV is not exhaustive; it focuses solely on the system's actual execution. This means that a violation of expected behaviour is only detected if it occurs in the execution trace. Despite this, RV is a lightweight technique because it does not check all possible system behaviours, allowing it to scale better than static verification methods, which often struggle with the state space explosion problem.

RV emerged after static verification methods like model checking~\cite{clarke1997model}, inheriting much from them, particularly in specifying the formal properties to verify. One of the most widely used formalisms in model checking, and consequently in RV, is Linear Temporal Logic (LTL)~\cite{DBLP:conf/focs/Pnueli77}. While we will detail its syntax and semantics later in the paper, we focus on LTL's implicit assumption of perfect information about the system. Typically, LTL verification assumes that the system under analysis provides all necessary information for verification~\cite{10.1145/2000799.2000800}. This is reflected at the verification level by generating atomic propositions that denote the knowledge about the system, which are used to verify the properties of interest. 
However, this assumption does not always hold, particularly for systems with autonomous, distributed, or faulty components, such as faulty sensors in real-world environments. In such cases, assuming all necessary information is available is overly optimistic. Although other works have addressed LTL RV with imperfect information~\cite{DBLP:conf/rv/StollerBSGHSZ11,DBLP:conf/rv/BartocciGKSSZS12,DBLP:conf/rv/KalajdzicBSSG13,DBLP:journals/corr/BartocciG13,DBLP:conf/rv/LeuckerSS0T19}, this research line is the first to tackle the problem fundamentally without requiring a new verification pipeline. Specifically, in~\cite{DBLP:conf/sefm/FerrandoM22}, we presented an initial attempt to extend the standard monitor synthesis pipeline to explicitly account for imperfect information. Building on this idea, we further refine our approach to incorporate the monitor's ability to be rational, enhancing its capability to handle imperfect information. While our previous work focused on engineering a monitor to provide correct answers despite imperfect information, our current work refines the monitor's information processing to yield conclusive results.

The contribution of this paper is twofold. First, we formally define the notion of imperfect information with respect to the monitor's visibility over the system, and re-engineer the LTL monitor's synthesis pipeline to recognise this visibility information. Second, we introduce the concept of a rational monitor, which can dynamically manage its visibility. To the best of our knowledge, this notion has not been previously defined in the literature. Additionally, we present the details of the prototype implemented to support our claims, providing the community with an LTL monitoring library that natively supports imperfect information and rationality. Finally, we show the use of this prototype in a case study.

The paper's structure is as follows. Section~\ref{sec:preliminaries} reports preliminaries notions.
Section~\ref{sec:case-study} introduces our case study.
Section~\ref{sec:rv-imperfect-info} formally presents the notion of imperfect information, its implication at the monitoring level and the resulting re-engineering of the standard LTL monitor's synthesis pipeline. 
Furthermore, Section~\ref{sec:rv-rational} provides the main tools to introduce monitor rationality.
Section~\ref{sec:implementation} reports the details on the prototype that has been developed as a result of the re-engineering process, along with some experiments of its use in a realistic case study. Section~\ref{sec:related-work} positions the paper against the state of the art. Section~\ref{sec:conclusions-and-future-work} concludes the paper and discusses some possible future directions.


\section{Preliminaries}
\label{sec:preliminaries}


A system $\system$ has an \textit{alphabet} $\eventSet$ with which it is possible to define the set $2^\eventSet$ of all its events. Given an alphabet $\eventSet$, a \emph{trace} $\trace=ev_0 ev_1 \ldots$, is a sequence of events in $2^\eventSet$. With $\trace(i)$ we denote the i-th element of $\trace$ (\textit{i.e.}, $ev_i$), $\trace^i$ the suffix of $\trace$ starting from $i$ (\textit{i.e.}, $ev_i ev_{i+1} \ldots$), $(2^{\eventSet})^*$ the \emph{set of all possible finite traces} over $\eventSet$, and $(2^{\eventSet})^\omega$ the \emph{set of all possible infinite traces} over $\eventSet$.

The standard formalism to specify properties in RV is Linear Temporal Logic (LTL~\cite{DBLP:conf/focs/Pnueli77}). 
The relevant parts of the syntax of LTL are the following:
{\small
$$
\spec := p \;|\; \lnot\spec \;|\; (\spec\lor\spec) \;|\; \nextOp\spec \;|\; (\spec \;\until\; \spec)  
$$
}
where $p\in\eventSet$ is an atomic proposition (aka atom), $\spec$ is a formula, $\Circle$ stands for \emph{next-time}, and $\until$ stands for \emph{until}. In the rest of the paper, we also use the standard derived operators, such as $(\spec\rightarrow\spec')$ instead of $(\lnot\spec\lor\spec')$, $\spec \;R\; \spec'$ instead of $\lnot(\lnot\spec \until \lnot\spec')$, $\square \spec$ (\emph{always} $\spec$) instead of ($false \;R\; \spec$), and $\lozenge \spec$ (\emph{eventually} $\spec$) instead of ($true \until \spec$).

Let $\trace\in (2^\eventSet)^\omega$ be an infinite sequence of events over $\eventSet$, the semantics of LTL is as follows:
\begin{eqnarray*}
\trace &\models& p \textrm{ if } p \in \trace(0)\\
\trace &\models& \lnot\spec \textrm{ if } \trace \not\models \spec\\
\trace &\models& \spec\lor\spec' \textrm{ if } \trace \models \spec \textrm{ or } \trace \models \spec'\\
\trace &\models& \Circle\spec \textrm{ if } \trace^1\models\spec\\
\trace &\models& \spec  \until \spec' \textrm{ if } \exists_{i \geq 0}.\trace^i\models\spec' \textrm{ and } \forall_{0 \leq j < i}.\trace^j\models\spec
\end{eqnarray*}
A trace $\trace$ satisfies an atomic proposition $p$, if $p$ belongs to $\trace(0)$; which means, $p$ holds in the initial event of the trace $\trace$. A trace $\trace$ satisfies the negation of the LTL property $\spec$, if $\trace$ does not satisfy $\spec$. 
A trace $\trace$ satisfies the disjunction of two LTL properties, if $\trace$ satisfies at least one of them. 
A trace $\trace$ satisfies next-time $\spec$, if the suffix of $\trace$ starting in the next step ($\trace^1$) satisfies $\spec$. Finally, a trace $\trace$ satisfies $\spec\until\spec'$, if there exists a suffix of $\trace$ s.t. $\spec'$ is satisfied, and for all suffixes before it, $\spec$ holds. 
Thus, given an LTL property $\spec$, we denote $\sem{\spec}$ the language of the property, \textit{i.e.}, the set of traces which satisfy $\spec$; namely $\sem{\spec} = \{ \trace \;|\; \trace \models \spec\}$.

In Definition~\ref{rv-def}, we present a general and formalism-agnostic definition of a monitor. Informally, a monitor is a function that, given a trace of events in input, returns a verdict which denotes the satisfaction (resp., violation) of a formal property over the trace.

\begin{definition}[Monitor]\label{rv-def}
Let~$\system$ be a system with alphabet~$\eventSet$, $\trace$ a finite trace,
and~$\spec$ be an LTL property. Then, a monitor for $\spec$ is a function
$\stmonitor{\spec}:(2^\eventSet)^*\rightarrow\mathbb{B}_3$, where
$\mathbb{B}_3=\{\top, \bot, \unknown \}$:
\vspace*{-0.25cm}
$$
\stmonitorAppl{\spec}{\trace} =
\left\{
\bgroup
\def\arraystretch{1.2}
  \begin{tabular}{cl}
  $\top$ & {\qquad$\forall_{\continuation \in (2^\eventSet)^\omega}.\trace \bullet \continuation \in \sem{\spec}$}\\
  $\bot$ & {\qquad$\forall_{\continuation \in (2^\eventSet)^\omega}.\trace \bullet \continuation \notin\sem{\spec}$}\\
  $\unknown$ & {\qquad$otherwise$}\\
  \end{tabular}
\egroup
\right.
$$

\vspace*{-0.25cm}
where $\bullet$ is the standard trace concatenation operator.
\end{definition}
Intuitively, a monitor returns~$\top$ if all continuations ($\continuation$) of $\trace$ satisfy $\spec$; $\bot$ if all possible
continuations of $\trace$ violate $\spec$; $\unknown$ otherwise. The first two outcomes are standard representations of satisfaction and violation, while the third is specific to RV. In more detail, it denotes when the monitor cannot conclude any verdict yet. This is closely related to the fact that RV is applied while the system is still running, and future events may still change the verdict. For instance, a property might be currently satisfied (resp., violated) by the system, but violated (resp., satisfied) in the (still unknown) future. The monitor can only safely conclude any of the two final verdicts ($\top$ or $\bot$) if it is sure such verdict will never change. The addition of the third outcome symbol $?$ helps the monitor to represent its position of uncertainty w.r.t. the current system execution.

A monitor function is usually implemented as a Finite State Machine (FSM), specifically a Moore machine (FSM where the output value of a state is only determined by the state)~\cite{DBLP:conf/fsttcs/BauerLS06,10.1145/2000799.2000800}. 
A Moore machine can be defined as a tuple $\langle Q,q_0,\Sigma,O,\delta,\gamma \rangle$, where $Q$ is a finite set of states, $q_0$ is the initial state, $\Sigma$ is the input alphabet, $O$ is the output alphabet, $\delta:Q\times\Sigma\rightarrow Q$ is the transition function mapping a state and an event to the next state, and $\gamma:Q\rightarrow O$ is the function mapping a state to the output alphabet.

In~\cite{10.1145/2000799.2000800}, Bauer \textit{et al}. present the sequence of steps required to generate from an LTL formula $\spec$ the corresponding Moore machine instantiating the $\stmonitor{\spec}$ function (as summarised in Figure~\ref{fsm-steps-fig}).

\vspace*{-1cm}
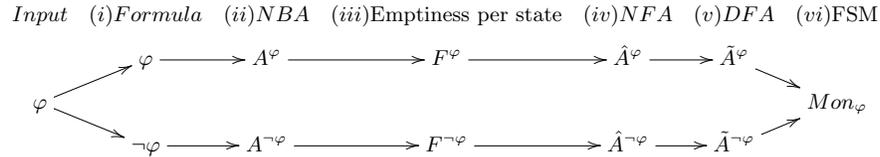
\begin{figure}[h!]
\[
\scalebox{0.90}{
\xymatrix @C=0.3em @R=0.3em{ 
{Input} & {(i)Formula} & {(ii)NBA} & {(iii) \textrm{Emptiness per state}} & {(iv) NFA} & {(v) DFA} & {(vi) \textrm{FSM}} \\
& {\spec} \ar[r] & {A^{\spec}} \ar[r] & {F^{\spec}} \ar[r] & {{\hat{A}}^{\spec}} \ar[r] & {{\tilde{A}}^{\spec}} \ar[rd] & \\
{\spec} \ar[ru] \ar[rd] & & & & & & {\stmonitor{\spec}} \\
& {\lnot \spec} \ar[r] & {A^{\lnot \spec}} \ar[r] & {F^{\lnot \spec}} \ar[r] & {{\hat{A}}^{\lnot \spec}} \ar[r] & {{\tilde{A}}^{\lnot \spec}} \ar[ru] & \\
}
}
\]
\caption{Steps required to generate an FSM from an LTL formula $\spec$. NBA is Non-deterministic B{\"u}chi Automaton, NFA is Non-deterministic Finite Automaton, and DFA is Deterministic Finite Automaton.}\label{fsm-steps-fig}
\end{figure}
\vspace*{-0.5cm}

Given an LTL property $\spec$, a series of transformations is performed on $\spec$, and its negation $\lnot\spec$. Considering $\spec$ in step \emph{(i)}, first, a corresponding NBA $A^\spec$ is generated in step \emph{(ii)}. This can be obtained using Gerth \textit{et al}.'s algorithm ~\cite{DBLP:conf/pstv/GerthPVW95}. Such automaton recognises the set of infinite traces that satisfy $\spec$ (according to LTL semantics). Then, each state of $A^\spec$ is evaluated; the states that when selected as initial states in $A^\spec$ do not generate the empty language are then added to the $F^\spec$ set in step \emph{(iii)}. With such a set, an NFA $\hat{A}^\spec$ is obtained from $A^\spec$ by simply substituting the final states of $A^\spec$ with $F^\spec$ in step \emph{(iv)}. $\hat{A}^\spec$ recognises the finite traces (prefixes) that have at least one infinite continuation satisfying $\spec$ (since the prefix reaches a state in $F^\spec$). After that, $\hat{A}^\spec$ is transformed (Rabin–Scott powerset construction~\cite{DBLP:journals/ibmrd/RabinS59}) into its equivalent deterministic version $\tilde{A}^\spec$ in step \emph{(v)}; this is possible since deterministic and non-deterministic finite automata have the same expressive power. The exact same steps are performed on $\lnot\spec$, which bring to the generation of the $\tilde{A}^{\lnot\spec}$ counterpart. The difference between $\tilde{A}^\spec$ and $\tilde{A}^{\lnot\spec}$ is that the former recognises finite traces which have continuations satisfying $\spec$, while the latter recognises finite traces which have continuations violating $\spec$. Finally, a Moore machine can be generated as a standard automata product between $\tilde{A}^\spec$ and $\tilde{A}^{\lnot\spec}$ in the final step \emph{(vi)}, where the states are denoted as tuples $(q,q')$, with $q$ and $q'$ belonging to $\tilde{A}^\spec$ and $\tilde{A}^{\lnot\spec}$, respectively. The outputs are then determined as: $\top$ if $q'$ does not belong to the final states of $\tilde{A}^{\lnot\spec}$, $\bot$ if $q$ does not belong to the final states of $\tilde{A}^{\spec}$, and $\unknown$ otherwise.
This brings us to the revised monitor construction as follows.

\begin{definition}[Monitor]\label{def:monitor-fsm}
Given an LTL formula $\spec$ and a finite trace $\sigma$, the revised monitor is defined as follows:
\vspace*{-0.25cm}
$$
\stmonitorAppl{\spec}{\trace} =
\left\{
\bgroup
\def\arraystretch{1.2}
  \begin{tabular}{cl}
  $\top$ & {\qquad$\trace\notin\lang{\tilde{A}^{\lnot\spec}}$}\\
  $\bot$ & {\qquad$\trace\notin\lang{\tilde{A}^{\spec}}$}\\
  $\unknown$ & {\qquad$\trace\in\lang{\tilde{A}^{\spec}} \land \trace\in\lang{\tilde{A}^{\lnot\spec}}$}\\
  \end{tabular}
\egroup
\right.
$$

\vspace*{-0.25cm}
where $\lang{A}$ denotes the language recognised by automaton $A$.
\end{definition}

Now that all the background notions have been introduced, we can move on with the more technical parts of the paper. Specifically, in Section~\ref{sec:rv-imperfect-info}, we present how to extend the notion of a monitor in case of imperfect information and, in Section~\ref{sec:rv-rational}, we show how to make a monitor rational in order to reason upon its lack of information. To facilitate the understanding of our technical contribution, we present our case study in the next section, which will serve as a running example in the rest of the paper.

\section{Remote inspection case study}
\label{sec:case-study}


Our case study is based on a 3D simulation of a Jackal\footnote{\url{https://clearpathrobotics.com/jackal-small-unmanned-ground-vehicle}}, a four-wheeled unmanned ground vehicle (referred to as the `rover' from now on), coupled with a simulated radiation sensor, that the rover uses to take radiation readings of points of interest while patrolling around a nuclear facility, and a camera, that the rover uses to inspect images of the nuclear waste barrels in the area. This simulation is based on the work presented in~\cite{robotics10030086}, which explains how the simulated sensor works and how radiation was simulated in the environment. In our version of the simulation the rover is autonomously controlled by a rational/intelligent agent~\cite{WooldridgeRao99:book}. Figure~\ref{fig:screenshot} reports a screenshot of the case study.

\begin{figure}[ht]
    \centering
    \includegraphics[width=0.9\linewidth]{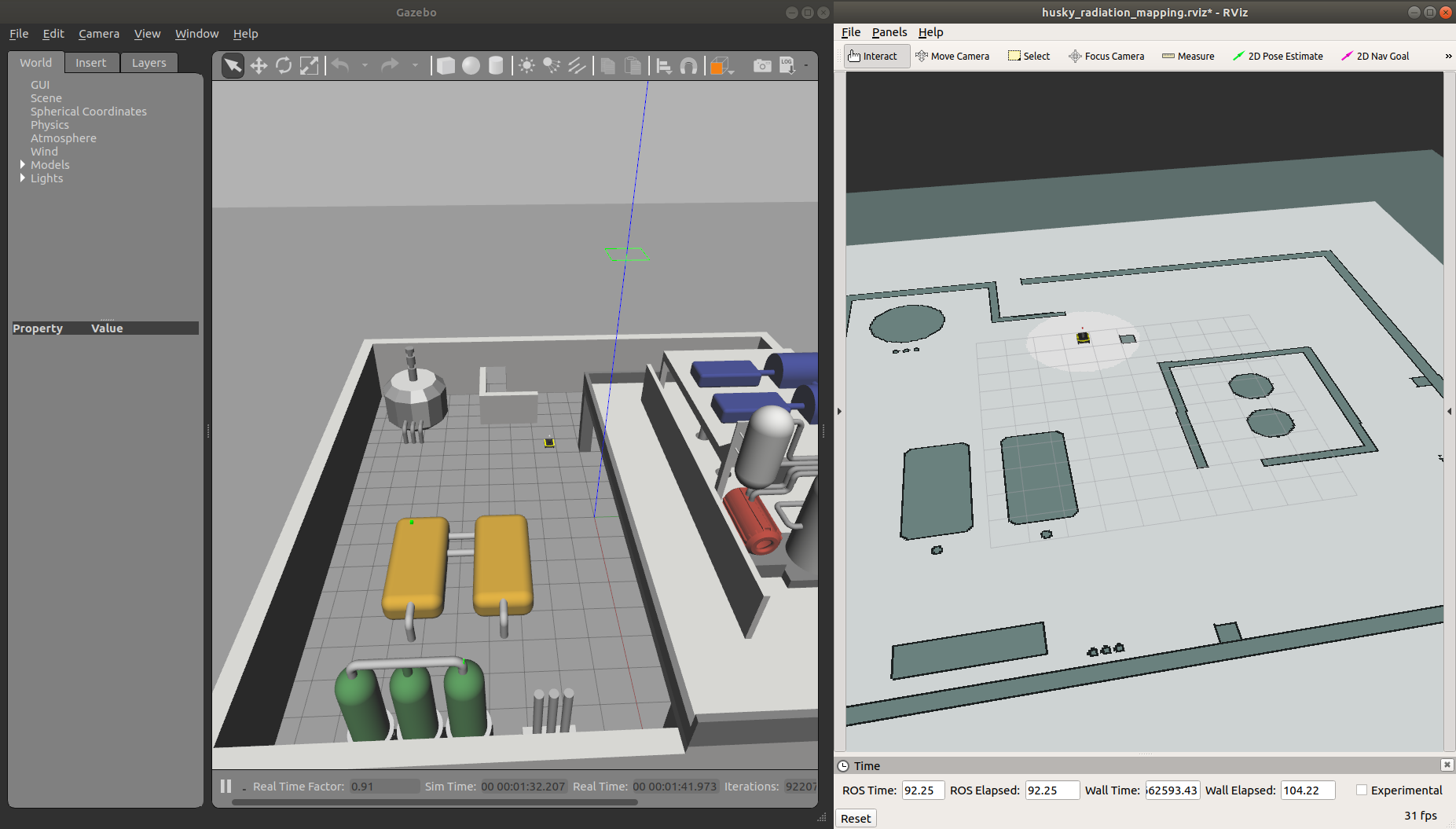}
    \caption{Simulation in Gazebo of the remote inspection of nuclear plant.}
    \label{fig:screenshot}
\end{figure}

A typical mission in our simulation starts with the rover positioned at the entrance of a nuclear facility. The goal of this mission is to inspect a number of points of interest (\textit{i.e.}, waypoints). Inspecting a waypoint serves two purposes: taking radiation readings to check if the radiation is at an acceptable level, and using a camera to detect abnormalities such as leakage in barrels and pipes. After inspecting all of the waypoints, the rover can either return to the entrance to await for a new mission, or keep patrolling and inspecting the waypoints.

Without losing generality, we assume the image captured by the rover's camera can be represented as a grid. Each cell in such a grid can contain, or not, an abnormality (e.g., a cut in the barrel). This information is translated into propositions, that can be transmitted to the monitor to be analysed at runtime.

We assume the presence of a cut on a barrel, denoted by $c$, and the presence of a rust stain, denoted by $s$. Furthermore, we assume the presence of $\alpha$, $\beta$, and $\gamma$ radiations. The act of moving back to base for decontamination is denoted by $mb$, and signalling a warning when a cut is found is denoted by $w$. Therefore, the set of atomic propositions is $\eventSet = \{b^1, b^2, b^3, c, s, \alpha, \beta, \gamma, mb, w\}$\footnote{Note that to simplify the case study, we assume that the rover only needs to analyse three ordered barrels: $b^1$, $b^2$, and $b^3$.}.

The first property we want to verify at runtime is whether the rover finds a cut ($c$) on a barrel and signals the presence of the cut through a warning ($w$). This property can be formulated as the following LTL formula:

$$\spec_1 = \lozenge (c \land \nextOp{w})$$

Furthermore, we are interested in checking whether the rover behaves properly in the presence of high levels of $\gamma$ radiation, consequently aborting the mission and going back to base for decontamination ($mb$). This formula can be formalised as:

$$\spec_2 = \lozenge (\gamma \land (b^1 \lor b^2 \lor b^3) \land \nextOp{mb})$$

Last, but not least, we have a property to check that when the rover inspects a barrel without finding a cut, it continues the inspection to the next barrel. The formula is formalised as follows:

$$\spec_3 = \lozenge ((\lnot c \land b^1 \land \nextOp{b^2}) \lor (\lnot c \land b^2 \land \nextOp{b^3}))$$

In addition, we have three safety objectives. The first concerns whether the rover will not find a cut in the barrel. This property can be formulated as the following LTL formula:

$$\psi_1 = \square ((b^1\lor b^2 \lor b^3) \Rightarrow \nextOp{\neg c})$$

Furthermore, another safety property of interest is to check whether any $\gamma$ radiation is detected away from any barrel, indicating a possible more widespread leakage not focused on a specific barrel. This property can be formalised as:

$$\psi_2 = \square (\gamma \Rightarrow \neg (b^1 \lor b^2 \lor b^3))$$

Last, but not least, we have a property to verify that in absence of $\gamma$ radiation the mission is not aborted. The formula follows:

$$\psi_3 = \square (\lnot\gamma \Rightarrow \lnot mb)$$

To better understand the properties, let us just assume that the trace of events $\trace$ observed by the rover is 
$$\sigma(0) = \{\}, \sigma(1) = \{b^1\},\sigma(2) = \{mb, b^2\}, \sigma(3) = \{\},\sigma(4) = \{w\}$$ 
Note that, when applied to $\sigma$ we have that $\spec_1$ is determined as $?$ by the standard LTL monitor. In fact, no events in $\sigma$ contains $c$. 
Similarly, in $\spec_2$, we once more obtain $?$ because no events in $\sigma$ contains $\gamma$.
Finally, $\spec_3$ is instead determined as $\top$ by the monitor. In fact, in $\sigma(1)$, $b^1$ holds, $c$ does not hold, and in $\sigma(2)$, $b^2$ holds.
On the safety properties side, $\psi_1$ is determined as $?$. This can be derived by the fact that in every event the property holds.
Concerning $\psi_2$, again the result is $?$, this is because the left operand of the implication never holds.
Finally, $\psi_3$ is determined as $\bot$ by the monitor. In fact, event $\sigma(2)$ contains $mb$, but not $\gamma$.

\section{Runtime Verification with Imperfect Information}
\label{sec:rv-imperfect-info}

So far, we have focused on standard RV of LTL properties. However, this standard approach relies on a critical assumption:

\begin{center}
    \textit{The absence of an atomic proposition is equivalent to \\ the negation of that proposition.}
\end{center}

This assumption holds true in formal verification for systems with perfect information -- that is, systems where every component has complete knowledge and visibility of the entire system. While this may be applicable to monolithic and traditional systems, it does not necessarily apply to autonomous, distributed, or artificial intelligence-driven systems. In such scenarios, having a complete view of the system is often not feasible.

Since RV relies on monitoring the system under analysis, if the verified component does not have complete access to the system's information, the monitor will also lack complete access. As a result, our runtime monitor may only observe partial information about the system. Consequently, the event traces provided to the monitor might be missing some atomic propositions, which could be incorrectly interpreted as the negation of those propositions. It is crucial to differentiate between knowing that something is not true and recognising that something is simply unknown.

\subsection{How can we formally represent the imperfect information?}

As previously discussed in this paper and in~\cite{DBLP:conf/sefm/FerrandoM22}, the issue with using LTL in systems with imperfect information lies in confusing the absence of an atomic proposition with its negation. When information is imperfect, the trace may lack certain atomic propositions that are simply unknown or unobservable. To address this, we need a method to explicitly characterise the absence of information. To achieve this, we adopt an approach similar to that in~\cite{DBLP:journals/ai/BelardinelliFM23,DBLP:journals/jair/BelardinelliLMY22}, where atomic propositions are duplicated.

One possible way to represent imperfect information is by allowing indistinguishability on atomic propositions $\eventSet$. To do this we introduce an equivalence relation $\sim$ over $\eventSet$ and its equivalence classes. 

\begin{definition}[Equivalence relation]\label{def:eq_relation}
    An equivalence relation $\sim\in\eventSet\times\eventSet$ determines what a monitor cannot distinguish.
    Specifically, given two atomic propositions $p,q \in \eventSet$, we say that they are indistinguishable if and only if $p \sim q$.
\end{definition}

\begin{definition}[Equivalence class]\label{def:eq_class}
    Given an equivalence relation $\sim$ we define with $\equivclass\in 2^\eventSet$ an equivalence class of $\sim$. Formally, given $p\in\equivclass$, for each $q\in\eventSet$, if $p \sim q$ then $q\in\equivclass$. Additionally, we define the witness of $\equivclass$ with the symbol $[\equivclass]$ and with $\Lambda$ the set of equivalence classes.
\end{definition}

To handle the verification process in the imperfect information context, we need to do some extensions. First of all, we can not simply use the set of atomic propositions $\eventSet$. In particular, we need to replace $\eventSet$ with a new set $\bar{\eventSet}$ that is defined as follows: 
for each $p \in \eventSet$ we have $p_{\top} \in \bar{\eventSet}$ and $p_{\bot} \in \bar{\eventSet}$.
That is, we duplicate the set of atomic proposition to make the truth value explicit. 


Without losing generality, we only consider LTL formulas in Negation Normal Form (NNF). An LTL formula in NNF has only negations at the atom levels (\textit{i.e.}, we only have $\lnot p$). Given an LTL formula, its NNF can be easily obtained by propagating all negations to the atoms. For instance, if we had $\lnot\nextOp p$, we would rewrite it as $\nextOp\lnot p$. The same goes for the other operators.

Now, we present how to generate the explicit version of an LTL formula.

\begin{definition}\label{def:explicit}
Given an LTL formula $\spec$ in NNF and the set of equivalence classes $\Lambda$, we define the explicit version of $\spec$ as follows:
\vspace*{-0.25cm}
\begin{eqnarray*}
\explicit{p} &=& [\lambda]_\top \\
\explicit{\lnot p} &=& [\lambda]_\bot \\
\explicit{\spec\lor\spec'} &=& \explicit{\spec}\lor\explicit{\spec'} \\
\explicit{\nextOp{\spec}} &=& \nextOp{\explicit{\spec}} \\
\explicit{\spec\until\spec'} &=& \explicit{\spec}\until\explicit{\spec'}
\end{eqnarray*}

\vspace*{-0.25cm}
where $\lambda \in \Lambda$ and $p \in \lambda$.
\end{definition}

We now present how to construct the explicit and visible versions of a trace.

\begin{definition}\label{def:track}
	Given a trace $\sigma$ and a set $\eventSet$, we define the explicit version of $\sigma$ as $\sigma_e$, for each element $\sigma(i)$ as follows:
	\vspace*{-0.25cm}
	\begin{itemize}
		\item for all $p \in \sigma(i)$, $p_\top \in \sigma_e(i)$;
		\item for all $p \in \eventSet \setminus \sigma(i)$, $p_\bot \in \sigma_e(i)$.
	\end{itemize}
\end{definition}

\begin{definition}\label{def:trackvis}
	Given an explicit trace $\sigma_e$ and the set of equivalence classes $\Lambda$, we define the visible version of $\sigma_e$ as $\sigma_v$, for each $\sigma(i)$ and $\lambda \in \Lambda$ as follows:
	\vspace*{-0.25cm}
	\begin{itemize}
		\item $[\lambda]_\top \in \sigma_v(i)$ if and only if for all $p \in \lambda$, $p_\top \in \sigma_e(i)$;
		\item $[\lambda]_\bot \in \sigma_v(i)$ if and only if for all $p \in \lambda$, $p_\bot \in \sigma_e(i)$.
	\end{itemize}
\end{definition}

For simplicity, in the rest of the paper, in the trivial case where an equivalence class consists of a single atomic proposition, we omit the square brackets.

Given the above elements, we define a three-valued semantics $\models_3$ for LTL:
\begin{eqnarray*}
(\trace &\models_3& p) = \top \textrm{ if } p_\top \in \trace(0)\\
(\trace &\models_3& p) = \bot \textrm{ if } p_\bot \in \trace(0)\\
(\trace &\models_3& \lnot\spec) = \top \textrm{ if } (\trace \not\models_3 \spec) = \top\\
(\trace &\models_3& \lnot\spec) = \bot \textrm{ if } (\trace \not\models_3 \spec) = \bot \\
(\trace &\models_3& \spec\lor\spec') = \top \textrm{ if } (\trace \models_3 \spec) = \top  \textrm{ or } (\trace \models_3 \spec') = \top \\
(\trace &\models_3& \spec\lor\spec') = \bot \textrm{ if } (\trace \models_3 \spec) = \bot  \textrm{ and } (\trace \models_3 \spec') = \bot \\
(\trace &\models_3& \Circle\spec) = \top \textrm{ if } (\trace^1\models_3\spec) = \top\\
(\trace &\models_3& \Circle\spec) = \bot \textrm{ if } (\trace^1\models_3\spec) = \bot\\
(\trace &\models_3& \spec  \until \spec') = \top \textrm{ if } \exists_{i \geq 0}.(\trace^i\models_3\spec') = \top \textrm{ and } \forall_{0 \leq j <i}.(\trace^j\models_3\spec) = \top \\
(\trace &\models_3& \spec  \until \spec') = \bot \textrm{ if } \forall_{i \geq 0}.(\trace^i\models_3\spec') = \bot \textrm{ or } \exists_{0 \leq j <i}.(\trace^j\models_3\spec) = \bot
\end{eqnarray*}

In all the other cases the truth value is undefined ($uu$).

To help the reader, we conclude the section with the following example.

\begin{example}\label{ex:imp_1}
Considering the rover example presented in Section~\ref{sec:case-study}, let us assume that, the rover is not always capable of sending and detecting all the information to the monitor. 
Because of this, the monitor is sometimes unable to distinguish between a cut and a rust stain, as well as between different kinds of radiation.
In such cases, from the viewpoint of the monitor analysing the scene, there is imperfect information over the atomic propositions. Then, we have imperfect information over $c$ and $s$, which is formalised as $c \sim s$ (\textit{i.e.}, there is an equivalence class $\lambda_{cs}$ between $c$ and $s$), and over $\alpha$, $\beta$, and $\gamma$, which is formalised as $\alpha \sim \beta \sim \gamma$ (\textit{i.e.}, there is an equivalence class $\lambda_{\alpha\beta\gamma}$ between $\alpha$, $\beta$, and $\gamma$).
Since in this scenario we have an equivalence relation between $c$ and $s$ (\textit{i.e.}, $c \sim s$) and between $\alpha$, $\beta$, and $\gamma$ (\textit{i.e.}, $\alpha \sim \beta \sim \gamma$), first we need to explicit the atomic propositions inside the formula, obtaining: 
$$\explicit{\spec_1} = \lozenge ([\lambda_{cs}]_\top \land \nextOp{w_\top})$$ 
$$\explicit{\spec_2} = \lozenge ([\lambda_{\alpha\beta\gamma}]_\top \land (b^1_\top \lor b^2_\top \lor b^3_\top))$$
$$\explicit{\spec_3} = \lozenge (([\lambda_{cs}]_\bot \land b^1_\top \land \nextOp{b^2_\top}) \lor ([\lambda_{cs}]_\bot \land b^2_\top \land \nextOp{b^3_\top}))$$
$$\explicit{\psi_1} = \square ((b^1_\top \lor b^2_\top \lor b^3_\top) \Rightarrow \nextOp{[\lambda_{cs}]_\bot})$$
$$\explicit{\psi_2} = \square ([\lambda_{\alpha\beta\gamma}]_\top \Rightarrow (b^1_\bot \land b^2_\bot \land b^3_\bot))$$
$$\explicit{\psi_3} = \square ([\lambda_{\alpha\beta\gamma}]_\bot \Rightarrow mb_\bot)$$

By using the newly updated LTL formula, we can generate the three automata as shown in Figure~\ref{fsm-steps-fig-ex}.

Note that, in Section~\ref{sec:case-study}, we assumed, as usual in RV, that the monitor had perfect information over the system. Thus, the trace $\trace$ there presented was assumed to perfectly denote the actual events generated by the system execution and observed by the monitor. However, that could not be the case in general, since as we mentioned before, the rover may not have the ability to send/detect all the information to the monitor, causing imperfect information in the latter. Let us now assume that the trace $\trace$ with perfect information was: 
$$\sigma(0) = \{\}, \sigma(1) = \{\gamma,b^1,c\}, \sigma(2) = \{\gamma, c, mb, b^2\}, \sigma(3) = \{c\}, \sigma(4) = \{w\}$$

Thus, we can update the trace of events of Section~\ref{sec:case-study} as well, first by generating its explicit version $\trace_e$ (see Definition~\ref{def:track}), where
$$\trace_e(0) = \{b^1_\bot,b^2_\bot,b^3_\bot, c_\bot, s_\bot, \alpha_\bot, \beta_\bot, \gamma_\bot, mb_\bot, w_\bot\}$$
$$\trace_e(1) = \{b^1_\top,b^2_\bot,b^3_\bot, c_\top, s_\bot, \alpha_\bot, \beta_\bot, \gamma_\top, mb_\bot, w_\bot\}
$$
$$\trace_e(2) = \{b^1_\bot,b^2_\top,b^3_\bot, c_\top, s_\bot, \alpha_\bot, \beta_\bot, \gamma_\top, mb_\top, w_\bot\}$$
$$\trace_e(3) = \{b^1_\bot,b^2_\bot,b^3_\bot, c_\top, s_\bot, \alpha_\bot, \beta_\bot, \gamma_\bot, mb_\bot, w_\bot\}$$
$$\trace_e(4) = \{b^1_\bot,b^2_\bot,b^3_\bot, c_\bot, s_\bot, \alpha_\bot, \beta_\bot, \gamma_\bot, mb_\bot, w_\top\}$$ 

Then by defining its visible version according to the given equivalence classes $\lambda_{cs}$ and $\lambda_{\alpha\beta\gamma}$ (see Definition~\ref{def:trackvis}), we obtain $\trace_v$, where 
$$\trace_v(0) = \{ b^1_\bot,b^2_\bot,b^3_\bot, [\lambda_{cs}]_\bot, [\lambda_{\alpha\beta\gamma}]_\bot, mb_\bot, w_\bot \}$$
$$\trace_v(1) = \{ b^1_\top,b^2_\bot,b^3_\bot, mb_\bot, w_\bot \}$$ 
$$\trace_v(2) = \{ b^1_\bot,b^2_\top,b^3_\bot, mb_\top, w_\bot \}$$ 
$$\trace_v(3) = \{ b^1_\bot,b^2_\bot,b^3_\bot, [\lambda_{\alpha\beta\gamma}]_\bot, mb_\bot, w_\bot  \}$$ 
$$\trace_v(4) = \{ b^1_\bot,b^2_\bot,b^3_\bot, [\lambda_{cs}]_\bot, [\lambda_{\alpha\beta\gamma}]_\bot, mb_\bot, w_\top  \}$$ 
Note that, as expected, the second and third events in $\trace_v$ do not contain information about the atomic propositions $c$ and $\gamma$. Furthermore, there is no information on the atomic proposition $c$ in the fourth event. 
This is determined by the fact that the atomic propositions $c_\top$ and $s_\bot$ hold in the second, third, and fourth events of $\trace_v$, and according to Definition~\ref{def:trackvis}, since $c \sim s$, we can have $[\lambda_{cs}]_\top$ (resp., $[\lambda_{cs}]_\bot$) if and only if both $c_\top$ and $s_\top$ hold (resp., $c_\bot$ and $s_\bot$). Thus, having a mismatch between the two atomic propositions (\textit{i.e.}, one is true while the other is false), we cannot safely add any witness for the equivalence class $\lambda_{cs}$.  
Instead, in the first and last events of $\trace_v$, since we have both $c_\bot$ and $s_\bot$, we can safely add the witness $[\lambda_{cs}]_\bot$ to the trace.
The same reasoning follows for the atomic proposition $\gamma$ and its equivalence class.
\end{example}



\subsection{Re-engineering Monitor with imperfect information}


Given an LTL formula and a visible trace for the monitor, we need a method to use them for performing RV. This can be achieved by extending the standard pipeline for generating LTL monitors (see Figure~\ref{fsm-steps-fig}). This extension involves two specific modifications:
\begin{enumerate}
    \item We use the explicit version of the LTL formula, following Definition~\ref{def:explicit}.
    \item We modify the product between $\tilde{A}^\spec$ and $\tilde{A}^{\lnot\spec}$ to generate the Moore machine representing the monitor.
\end{enumerate}

The resulting extension is illustrated in Figure~\ref{fsm-steps-fig-ex}. In this figure, the explicit version of the LTL formula is generated in step (ii), and the updated product between the automata is obtained in step (vii). The remaining steps are unchanged compared to Figure~\ref{fsm-steps-fig}
%
with the exception that the atomic propositions in the formula are duplicated before using the formula to generate the corresponding NBA, and an additional automaton has been added. This duplication of atomic propositions is crucial, as it completely changes the semantics of the subsequent steps in the monitor synthesis pipeline. Specifically, it is not true that for any given visible trace $\trace_v$, we have $\trace_v\notin\lang{\hat{A}^\spec}\Rightarrow\trace_v\in\lang{\hat{A}^{\lnot\spec}}$, nor $\trace_v\notin\lang{\hat{A}^{\lnot\spec}}\Rightarrow\trace_v\in\lang{\hat{A}^{\spec}}$. 
This means that when a visible trace of events $\trace_v$ is not a good prefix for $\spec$ (\textit{i.e.}, a prefix that can be extended to an infinite trace satisfying $\spec$), it does not necessarily have to be a bad prefix for $\spec$ (\textit{i.e.}, a prefix that cannot be extended to an infinite trace satisfying $\spec$). This aspect is closely related to the reason why a third formula ($\otimes\spec$) has been introduced in Figure~\ref{fsm-steps-fig-ex}. Since duplicating the atomic propositions in the formula breaks the duality between $\spec$ and $\lnot\spec$, a third automaton ($\tilde{A}^{\otimes\spec}$) is needed to recognise all the traces that neither satisfy nor violate $\spec$.
For this reason, we extended the pipeline by adding $\otimes\spec$, which is an abbreviation for $\lnot\explicit{\spec}\land\lnot\explicit{\lnot\spec}$. The automaton $\tilde{A}^{\otimes\spec}$, obtained by following the same steps as for the positive $\tilde{A}^{\explicit{\spec}}$ and negative $\tilde{A}^{\explicit{\lnot\spec}}$ automata, recognises all prefixes for which no continuation satisfying or violating $\spec$ exists.

\begin{center}
\begin{figure}[ht]
\[
\scalebox{0.8}{
\xymatrix @C=0.3em @R=1.0em{ 
{Input} & {(i)Formula} & {(ii)Explicit} & {(iii)NBA} & {(iv) \textrm{Emptiness per state}} & {(v) NFA} & {(vi) DFA} & {(vii) \textrm{FSM}} \\
& {\spec} \ar[r] & \explicit{\spec} \ar[d] \ar[r] & {A^{\explicit{\spec}}} \ar[r] & {F^{\explicit{\spec}}} \ar[r] & {{\hat{A}}^{\explicit{\spec}}} \ar[r] 
& {{\tilde{A}}^{\explicit{\spec}}} \ar[rd] & \\
{\spec} \ar[ru] \ar[rd] & & \otimes\spec \ar[r]
& A^{\otimes\spec} \ar[r] & F^{\otimes\spec} \ar[r] & \hat{A}^{\otimes\spec} \ar[r] & \tilde{A}^{\otimes\spec} \ar[r] & {\stmonitor{\spec}} \\
& {\lnot \spec} \ar[r] & \explicit{\lnot \spec} \ar[u] \ar[r] & {A^{\explicit{\lnot \spec}}} \ar[r] & {F^{\explicit{\lnot \spec}}} \ar[r] & {{\hat{A}}^{\explicit{\lnot \spec}}} \ar[r] 
& {{\tilde{A}}^{\explicit{\lnot \spec}}} \ar[ru] & \\
}
}
\]
\caption{Extended pipeline to consider imperfect information.}\label{fsm-steps-fig-ex}
\end{figure}
\end{center}
Now, we formalise the above reasoning with the following lemma.

\begin{lemma}
	Given a visible finite trace $\trace_v$ and an LTL formula $\varphi$, we have:
	\begin{center}
		$\trace_v\not\in\lang{\hat{A}^{\explicit{\spec}}} \not \Rightarrow \trace_v\in\lang{\hat{A}^{\explicit{\lnot\spec}}}$ \\
		$ \trace_v\not\in\lang{\hat{A}^{\explicit{\lnot\spec}}} \not \Rightarrow \trace_v\in\lang{\hat{A}^{\explicit{\spec}}}$
	\end{center}
\end{lemma}

\begin{proof}
	Assume we have a visible trace $\trace_v$ that is not included in the NFA $\hat{A}^{\explicit{\spec}}$. To prove our result, we need to show that $\trace_v$ is also not included in $\hat{A}^{\explicit{\lnot\spec}}$. To demonstrate this, suppose $\eventSet = \{p, q, r\}$, $\varphi = \nextOp{p}$, $p \sim q$, and $\trace$ where $\trace(1) = \{p\}$. Given Definitions \ref{def:track} and \ref{def:trackvis}, we can conclude that $\trace_v(1) = \{r_\bot\}$. Therefore, $\trace_v$ does not satisfy $\varphi$ and, consequently, is not included in the NFA $\hat{A}^{\explicit{\spec}}$. However, it is also not included in $\hat{A}^{\explicit{\lnot\spec}}$. This is because $p_\top$ and $p_\bot$ are not included in $\trace_v(1)$. This concludes the first relation. For the second relation, we can use a variant of the above reasoning.
\end{proof}



By adding the third automaton, the corresponding FSM synthesis also needs to change. The revised version is detailed in the following definition. 

\begin{definition}[Monitor with imperfect information]\label{def:monitor-imperfect}
Given an LTL formula $\spec$ and a visible trace $\trace_v$, a monitor with imperfect information is so defined:
$$
\stmonitorApplv{\spec}{\trace_v} =
\left\{
\bgroup
\def\arraystretch{1.2}
  \begin{tabular}{cl}
  $\top$ & {\qquad$ \trace_v\in\lang{\tilde{A}^{\explicit{\spec}}} \land \trace_v\notin\lang{\tilde{A}^{\explicit{\lnot\spec}}}\land \trace_v\notin\lang{\tilde{A}^{\otimes\spec}}$}\\
  $\bot$ & {\qquad$\trace_v\notin\lang{\tilde{A}^{\explicit{\spec}}}\land \trace_v\in\lang{\tilde{A}^{\explicit{\lnot\spec}}} \land \trace_v\notin\lang{\tilde{A}^{\otimes\spec}}$}\\
  $uu$ & {\qquad$\trace_v\notin\lang{\tilde{A}^{\explicit{\spec}}}\land\trace_v\notin\lang{\tilde{A}^{\explicit{\lnot\spec}}} \land \trace_v\in\lang{\tilde{A}^{\otimes\spec}}$}\\
  $?_{\not\bot}$ & {\qquad$\trace_v\in\lang{\tilde{A}^{\explicit{\spec}}} \land \trace_v\notin\lang{\tilde{A}^{\explicit{\lnot\spec}}}\land \trace_v\in\lang{\tilde{A}^{\otimes\spec}}$}\\
  $?_{\not\top}$ & {\qquad$\trace_v\notin\lang{\tilde{A}^{\explicit{\spec}}}\land\trace_v\in\lang{\tilde{A}^{\explicit{\lnot\spec}}}\land\trace_v\in\lang{\tilde{A}^{\otimes\spec}}$}\\
  $\unknown$ & {\qquad$\trace_v\in\lang{\tilde{A}^{\explicit{\spec}}} \land \trace_v\in\lang{\tilde{A}^{\explicit{\lnot\spec}}} \land \trace_v\in\lang{\tilde{A}^{\otimes\spec}}$}\\
  \end{tabular}
\egroup
\right.
$$
\end{definition}

In this definition, we see how the inclusion of a third automaton in the equation allows us to synthesise a finer monitor, in terms of the number of possible outcomes it returns. Compared to Definition~\ref{def:monitor-fsm}, we now have three additional outcomes. Specifically, given a visible trace $\trace_v$, the monitor returns $\top$ if there is no continuation of $\trace_v$ that either violates $\explicit{\spec}$ or makes it undefined. Conversely, it returns $\bot$ if there is no continuation that either satisfies $\explicit{\spec}$ or makes it undefined. With three automata, there is an additional final outcome to consider, which is $uu$. Thus, the monitor returns $uu$ if there is no continuation that either satisfies or violates $\explicit{\spec}$. These first three outcomes derive from the three-valued semantics for LTL. 
Additionally, we may encounter $?_{\not\bot}$, which is read as ``unknown, but it will never be violated from the monitor's point of view''. This outcome is returned when the visible trace $\trace_v$ has no continuation that will eventually violate $\explicit{\spec}$, but there are continuations that satisfy $\explicit{\spec}$ and make it undefined. Symmetrically, we have $?_{\not\top}$, which is read as ``unknown, but it will never be satisfied from the monitor's point of view''. This outcome is the dual of the previous one, where no continuations satisfying $\explicit{\spec}$ can be found, but continuations that violate $\explicit{\spec}$ and make it undefined exist. Lastly, we may encounter $\unknown$, denoting the completely unknown case. This outcome concerns situations where the monitor cannot yet conclude anything, as there exist continuations satisfying $\explicit{\spec}$, continuations violating $\explicit{\spec}$, and continuations that make it undefined.
 
\begin{remark}
Note that, in Definition~\ref{def:monitor-imperfect}, not all possible combinations are included. Specifically, it is not possible to have $\trace_v \notin \lang{\tilde{A}^{\explicit{\spec}}} \land \trace_v \notin \lang{\tilde{A}^{\explicit{\lnot\spec}}} \land \trace_v \notin \lang{\tilde{A}^{\otimes\spec}}$ and $\trace_v \in \lang{\tilde{A}^{\explicit{\spec}}} \land \trace_v \in \lang{\tilde{A}^{\explicit{\lnot\spec}}} \land \trace_v \notin \lang{\tilde{A}^{\otimes\spec}}$. The former is not possible because, by the definition of the three-valued semantics for LTL, there exists at least one automaton that includes the trace. The latter follows from the fact that it is unfeasible, given the nature of a visible trace, for a formula to be both true and false but not undefined in the future.
\end{remark}

In what follows, we provide two preservation results from the monitor with imperfect information to the one with perfect information.

\begin{lemma}
	Given a finite trace $\sigma$, a monitor with its visibility $\stmonitorApplv{\spec}{\trace}$, and a general monitor $\stmonitorAppl{\spec}{\trace}$, we have that:
	\begin{align}
		\text{if } \stmonitorApplv{\spec}{\trace_v} = \top & \text{ then } \stmonitorAppl{\spec}{\trace} = \top \\
		\text{if } \stmonitorApplv{\spec}{\trace_v} = \bot & \text{ then } \stmonitorAppl{\spec}{\trace} = \bot
	\end{align}
\end{lemma}

\begin{proof}
$\\$
$\\$
(1)	Suppose $\stmonitorApplv{\spec}{\trace_v} = \top$. This means that the visible trace $\trace_v$ satisfies the formula $\explicit{\varphi}$. We want to prove that the original trace $\trace$ satisfies the formula $\varphi$. To do this, given $\trace_v$, by Definitions \ref{def:track} and \ref{def:trackvis}, we know that for each $\trace_v(i)$, for all $p_\top \in \trace_v(i)$, $p \in \trace(i)$, and for all $p_\bot \in \trace_v(i)$, $p \notin \trace(i)$. Given the above reasoning, we need to provide an induction proof over the structure of the formula $\explicit{\varphi}$.

\textbf{Case:} $\explicit{\varphi} = p_\top$. So, $\varphi = p$. By hypothesis, $\stmonitorApplv{\spec}{\trace_v} = \top$. By the semantics of three-valued LTL, this means that $p_\top \in \trace_v(0)$, and by Definitions \ref{def:track} and \ref{def:trackvis}, $p \in \trace(0)$. Therefore, $\stmonitorAppl{\spec}{\trace} = \top$.

\textbf{Case:} $\explicit{\varphi} = p_\bot$. Thus, $\varphi = \neg p$. By hypothesis, $\stmonitorApplv{\spec}{\trace_v} = \top$. By the semantics of three-valued LTL, this means that $p_\bot \in \trace_v(0)$, and by Definitions \ref{def:track} and \ref{def:trackvis}, $p \notin \trace(0)$. Therefore, $\stmonitorAppl{\spec}{\trace} = \top$.

Since in the inductive cases the transformation of Definition \ref{def:explicit} does not change the structure and elements of the formula, we can conclude the proof.\\

\noindent (2) Suppose $\stmonitorApplv{\spec}{\trace_v} = \bot$. This means that the visible trace $\trace_v$ does not satisfy the formula $\explicit{\varphi}$. We want to prove that the original trace $\trace$ does the same for the formula $\varphi$. As in the previous case, we need to prove the implication by induction over the structure of the formula $\explicit{\varphi}$ for the base cases.

\textbf{Case:} $\explicit{\varphi} = p_\top$. So, $\varphi = p$. By hypothesis, $\stmonitorApplv{\spec}{\trace_v} = \bot$. By the semantics of three-valued LTL, this means that $p_\bot \in \trace_v(0)$, and by Definitions \ref{def:track} and \ref{def:trackvis}, $p \notin \trace(0)$. Therefore, $\stmonitorAppl{\spec}{\trace} = \bot$.

\textbf{Case:} $\explicit{\varphi} = p_\bot$. Thus, $\varphi = \neg p$. By hypothesis, $\stmonitorApplv{\spec}{\trace_v} = \bot$. By the semantics of three-valued LTL, this means that $p_\top \in \trace_v(0)$, and by Definitions \ref{def:track} and \ref{def:trackvis}, $p \in \trace(0)$. Therefore, $\stmonitorAppl{\spec}{\trace} = \bot$.

\end{proof}

Given the above results, we can deduce the following corollary.

\begin{corollary}
    Given a visible finite trace $\trace_v$ and an LTL formula $\varphi$, we have:
	\begin{center}
		$\trace_v\not\in\lang{\hat{A}^{\explicit{\spec}}} \Rightarrow \trace_v\in\lang{\hat{A}^{\explicit{\lnot\spec}}} \lor \trace_v\in\lang{\tilde{A}^{\otimes\spec}}$ \\
		$ \trace_v\notin\lang{\hat{A}^{\explicit{\lnot\spec}}} \Rightarrow \trace_v\in\lang{\hat{A}^{\explicit{\spec}}} \lor \trace_v\in\lang{\tilde{A}^{\otimes\spec}}$
		$\trace_v\notin\lang{\tilde{A}^{\otimes\spec}} \Rightarrow \trace_v\in\lang{\hat{A}^{\explicit{\lnot\spec}}} \lor \trace_v\in\lang{\hat{A}^{\explicit{\spec}}}$ \\
	\end{center}
\end{corollary}

\begin{example}\label{ex:imp_2}
Considering once more the running example presented in Section~\ref{sec:case-study} and Example~\ref{ex:imp_1}.
Thanks to our three-value semantics and the presence of explicit atomic propositions, the trace $\trace$ which was erroneously classifying $\spec_3$ as $\top$ and $\psi_3$ as $\bot$ from the standard LTL monitor before, now is classified as $?_{\not\bot}$ and $?_{\not\top}$, respectively. The semantics of the two verdicts is fundamentally different, as well as the reaction that the system should have. In the first case, by using a standard LTL monitor, the verdict returned by the monitor was $\top$. Thus, the agent controlling the rover could have used such information to continue the inspection with another barrel and not detecting a danger. In the second case, by using the extended LTL monitor that we presented in this work, the verdict returned by the monitor was $?_{\not\bot}$. Thus, the agent controlling the rover could use this information to, for instance, ask the rover to check again, maybe taking another picture. Even though this is a simple example, it allows us to show how our extension tackles the foundations of the imperfect information issue.
A complementary reasoning can be followed for $\psi_3$. The rest of the properties are not reported because the extended LTL monitor keeps returning an inconclusive verdict as the standard LTL monitor.
\end{example}



\section{Runtime Verification with Rational Monitor}
\label{sec:rv-rational}

In Section~\ref{sec:rv-imperfect-info}, we demonstrated how to engineer a monitor with imperfect information and specified it by extending the classical pipeline for creating a standard three-valued monitor. Note that, through this contribution, we enable the monitor to avoid incorrect truth values, ensuring the correctness of the truth value returned by the monitor. However, the monitor is ``passive'' concerning imperfect information. In this section, we will focus on improving the monitor's ability to handle imperfect information, making it ``active''. To do this, we introduce the concept of resources on the monitor's visibility. Specifically, the monitor is able to select the atoms that it wants to see w.r.t. its limited resources. We will introduce in detail this notion in Section~\ref{sec:resource}.
We further analyse the monitor's ability by introducing in Section~\ref{sec:dynamic} the notion of reactive monitor. To achieve this, we assume that the monitor can dynamically modify its visibility in different time windows. In Figure~\ref{fig:rational-monitor-hierarchy}, we provide a high-level overview, with a sort of class diagram, of the relationship between the rational, active, and reactive monitors.

\tikzset{every picture/.style={line width=0.75pt}} 

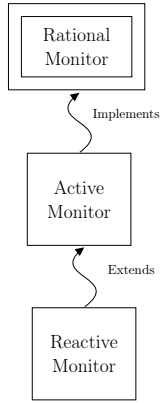
\begin{figure}
\centering
\scalebox{0.5}{

\begin{tikzpicture}[x=0.75pt,y=0.75pt,yscale=-1,xscale=1]

\draw   (27,24) -- (164.5,24) -- (164.5,111) -- (27,111) -- cycle(151.45,37.05) -- (40.05,37.05) -- (40.05,97.95) -- (151.45,97.95) -- cycle ;
\draw   (46,175) -- (150.5,175) -- (150.5,268) -- (46,268) -- cycle ;
\draw    (99,174) .. controls (138.2,144.6) and (61.67,144.01) .. (96.24,115.76) ;
\draw [shift={(98.5,114)}, rotate = 143.13] [fill={rgb, 255:red, 0; green, 0; blue, 0 }  ][line width=0.08]  [draw opacity=0] (8.93,-4.29) -- (0,0) -- (8.93,4.29) -- cycle    ;
\draw   (51,331) -- (155.5,331) -- (155.5,424) -- (51,424) -- cycle ;
\draw    (104,330) .. controls (143.2,300.6) and (66.67,300.01) .. (101.24,271.76) ;
\draw [shift={(103.5,270)}, rotate = 143.13] [fill={rgb, 255:red, 0; green, 0; blue, 0 }  ][line width=0.08]  [draw opacity=0] (8.93,-4.29) -- (0,0) -- (8.93,4.29) -- cycle    ;

\draw (95.75,67.5) node  [font=\Large] [align=left] {\begin{minipage}[lt]{60.7pt}\setlength\topsep{0pt}
\begin{center}
Rational \\Monitor
\end{center}

\end{minipage}};
\draw (145.75,136.5) node  [font=\normalsize] [align=left] {\begin{minipage}[lt]{55.47pt}\setlength\topsep{0pt}
\begin{center}
Implements
\end{center}

\end{minipage}};
\draw (98.25,221.5) node  [font=\Large] [align=left] {\begin{minipage}[lt]{51.69pt}\setlength\topsep{0pt}
\begin{center}
Active\\Monitor
\end{center}

\end{minipage}};
\draw (150.75,292.5) node  [font=\normalsize] [align=left] {\begin{minipage}[lt]{39.59pt}\setlength\topsep{0pt}
\begin{center}
Extends
\end{center}

\end{minipage}};
\draw (103.25,377.5) node  [font=\Large] [align=left] {\begin{minipage}[lt]{59.87pt}\setlength\topsep{0pt}
\begin{center}
Reactive\\Monitor
\end{center}

\end{minipage}};

\end{tikzpicture}

}
    \caption{Rational Monitor hierarchy.}
    \label{fig:rational-monitor-hierarchy}
\end{figure}



\subsection{How reasoning upon resources can help fighting imperfect information?}
\label{sec:resource}

As pointed out before, we introduce the notion of resources in order to grant the monitor the ability of reasoning upon its own visibility.
That is, when we claim that a monitor does not have visibility over a set of atomic propositions, this does not mean the monitor does not have a mean to access such information (in general), but, that by accessing such information may incur a certain cost.

Let us consider a robotic scenario where a robot operates in an environment and can only access information provided by its sensors. Such information can be incomplete. In fact, the robot may have limited energy consumption capabilities (\textit{i.e.}, it has a resource bound), making it unreasonable to access all sensors at once. This limitation affects how much power the robot can allocate to its various functions, such as accessing and operating its sensors. Consequently, the robot must manage its energy efficiently, which may involve prioritising certain sensors over others, thus causing incomplete information, to conserve energy and ensure longer operational periods.

In the rest of the section, we formalise the notion of resources and costs associated to the monitor.

Given a set of atomic propositions $\Sigma$ and an equivalence relation $\sim$, we can derive the set of equivalence classes $\equivset$ (according to Definition~\ref{def:eq_class}). 
Thus, we define $cost$: $\equivset \rightarrow \mathbb{N}$ that assign for each equivalence class $\equivclass \in \equivset$ a natural number $cost(\equivclass)$. The latter represents the cost for the monitor to make visible (or break) the atomic propositions involved in the equivalence class $\equivclass$. 

In this work, we assume our monitor has a limited number of resources and, by consequence, it could be unable to break all the equivalence classes.

So, given a cost function and the resource bound of the monitor, the latter needs to determine the best selection of the atomic propositions. To do so, we assume another function $\pay$: $\equivset \rightarrow \mathbb{N}$ that assigns for each $\equivclass$ a natural number that represents the expected payoff of breaking $\equivclass$. 

To generate a payoff, we need to consider the formula $\varphi$ under examination and the atomic propositions involved in it. To do this, we can define another function $metric_\varphi: \Sigma \rightarrow [0,1]$ that assigns a metric $metric_\varphi(p)$ to each atomic proposition. There are different approaches to producing the $\pay$ function. In the rest of the paper, we assume that $\pay(\equivclass) = \sum_{p \in \equivclass} metric_\varphi(p)$, meaning the payoff of $\equivclass$ is given by the sum of all the metrics of the atomic propositions involved in $\equivclass$.

To assign a metric for each atomic proposition, we need to consider the relevance of it in terms of the LTL formula under exam. In particular, we first need to assign a value for each syntactic element in the LTL syntax and then study the structure of the formula to determine the corresponding metrics. 
\begin{example}\label{ex:metric}
For instance, we can consider the following evaluation:
\begin{eqnarray*}
metric_q(p) &=&  0 \; with \; p \neq q\\
metric_p(p) &=&  1\\
metric_{\lnot p}(p) &=& 1 \\ 
metric_{\spec \lor \spec'}(p) &=& (metric_{\spec}(p) + metric_{\spec'}(p)) / 2 \\
metric_{\spec \land \spec'}(p) &=& max(metric_{\spec}(p), metric_{\spec'}(p))\\
metric_{\spec\Rightarrow\spec'}(p) &=& (metric_{\spec}(p) + metric_{\spec'}(p)) / 2 \\ 
metric_{\nextOp{\spec}}(p) &=& 0.5 \times metric_{\spec}(p) \\
metric_{\spec\until\spec'}(p) &=& 0.3 \times metric_{\spec}(p) + 0.7 \times metric_{\spec'}(p)\\
metric_{\spec\release\spec'}(p) &=& 0.3 \times metric_{\spec}(p) + 0.7 \times metric_{\spec'}(p)
\end{eqnarray*}
\end{example}

Given the $cost$ function and the $\pay$ function, our aim is to let the monitor able to select the best set of equivalence classes to break. To select such a subset we can use classic dynamic programming approaches like the one of the knapsack problem in which it is possible to minimise the costs while maximising the metrics.

Algorithm~\ref{alg:active_monitor} outlines the process for achieving active monitoring. Initially, the algorithm selects which indistinguishability relations to break (lines 2 and 3), leveraging the well-known Knapsack algorithm~\cite{cormen2022introduction} to balance the trade-off between the cost and the payoff of breaking an indistinguishability relation (line 2). This step can be executed through various methods, with the Knapsack algorithm being just one possible optimisation technique. Note that, since the Knapsack algorithm is performed on sets of atomic propositions (the equivalence classes), we define an order on the equivalence classes with the same payoff. That is, if two equivalence classes score the same payoff, then the Knapsack algorithm chooses to break the equivalence classes with the greater number of atomic propositions. In case the number of atomic propositions is the same, then the choice is random.
Subsequently, the algorithm updates the indistinguishability relation based on the Knapsack problem's output (line 3). 

Next, a monitor for the formula $\varphi$ is generated (line 4), and the trace $\sigma$ is updated in accordance with the new indistinguishability relation (lines 5--7). By the end of the loop, $\sigma$ represents the monitor's visible trace based on its visibility criteria (what we previously referred in the paper as $\sigma_v$). Finally, the monitor's outcome is returned (line 8).

{
\begin{algorithm}[t]
\caption{ActiveMonitor $\langle \sigma,\varphi,\sim,\pay, cost, \mathbf{b} \rangle$}
\label{alg:active_monitor}
\begin{algorithmic}[1]
\State{$count = 1$}
\State{$break = knapsack(\pay, cost, \mathbf{b})$}
\State{$\sim' \;= \;(\sim \setminus\; break)$}
\State{$Mon_\varphi = Monitor(\varphi)$}
\While{$count \leq |\sigma|$}
\State{$\sigma[count] = \sigma[count] \setminus \sim'$}
\State{$count = count + 1$}
\EndWhile
\State{\textbf{return } $Mon_\varphi(\sigma[1,count])$}
\end{algorithmic}
\end{algorithm}
}

\begin{example}\label{ex:active}
    Considering the running example presented in Section~\ref{sec:case-study} and revisited in Example~\ref{ex:imp_1} and Example~\ref{ex:imp_2}, we now show the impact of an active monitor. 
    Suppose that the cost of breaking the equivalence class $\equivclass_{cs}$ is $2$ and the cost of breaking the equivalence class $\equivclass_{\alpha\beta\gamma}$ is $3$.
    Given the metric of Example~\ref{ex:metric} and the formula $\spec_1$, we can determine the payoff of the atoms of the equivalence classes $\equivclass_{cs}$ and $\equivclass_{\alpha\beta\gamma}$ (\textit{i.e.}, $metric_{\spec_1}$). 
    In more detail, we obtain that the payoff to break the equivalence class $\equivclass_{cs}$ is $0.7$, which can be obtained through the computation: $\pay(\equivclass_{cs}) = metric_{\spec_1}(c) + metric_{\spec_1}(s)$, with $metric_{\spec_1}(c)=0.3 \times 0 + 0.7 \times max(1, 0.5 \times 0) = 0.7$, and $metric_{\spec_1}(s)= 0.3 \times 0 + 0.7 \times max(0, 0.5 \times 0) = 0.0$.
    With the same reasoning, we obtain that $\pay(\equivclass_{\alpha\beta\gamma})=0.0$, since no atomic proposition in $\equivclass_{\alpha\beta\gamma}$ is of interest for the verification of $\spec_1$.
    By assuming a bound $\textbf{b}$ greater or equal than $2$, the ActiveMonitor can break $\equivclass_{cs}$ and generate a new trace of events:
$$\trace_v(0) = \{ b^1_\bot,b^2_\bot,b^3_\bot, c_\bot,s_\bot, [\equivclass_{\alpha\beta\gamma}]_\bot, mb_\bot, w_\bot \}$$
$$\trace_v(1) = \{ b^1_\top,b^2_\bot,b^3_\bot,c_\top,s_\bot, mb_\bot, w_\bot \}$$ 
$$\trace_v(2) = \{ b^1_\bot,b^2_\top,b^3_\bot,c_\top,s_\bot, mb_\top, w_\bot \}$$ 
$$\trace_v(3) = \{ b^1_\bot,b^2_\bot,b^3_\bot,c_\top,s_\bot,[\equivclass_{\alpha\beta\gamma}]_\bot, mb_\bot, w_\bot  \}$$ 
$$\trace_v(4) = \{ b^1_\bot,b^2_\bot,b^3_\bot, c_\bot,s_\bot, [\equivclass_{\alpha\beta\gamma}]_\bot, mb_\bot, w_\top  \}$$
    Thanks to its ability, the ActiveMonitor can conclude with $\top$ for the liveness property $\spec_1$.
    This is determined by the fact that in $\trace_v(3)$ the atomic proposition $c_\top$ holds and in $\trace_v(4)$ the atomic proposition $w_\top$ holds.
    The same reasoning can be followed for the liveness property $\spec_2$, in which instead of breaking $\equivclass_{cs}$, the $\equivclass_{\alpha\beta\gamma}$ equivalence class is broken. Note that, to accomplish the task of breaking the latter equivalence class the ActiveMonitor needs to have $\textbf{b}$ greater or equal than $3$. 
    The same reasoning can be followed to handle the safety formulas $\psi_1$ and $\psi_2$ presented in Section~\ref{sec:case-study}.
\end{example}

\subsection{How can we formally represent the dynamic behaviour of the monitor?}
\label{sec:dynamic}

In the previous section, we have introduced the notion of active monitor. However, our goal is to introduce a monitor able to dynamically reason upon its imperfect information. Thus, in this section, we introduce the notion of ``reactive'' monitor.

First, we need to introduce the concept of a \emph{time window}. A time window allows us to split the input trace into different frames. We can assume that within each frame, the monitor's visibility is static, while it can change between frames, making it dynamic. At the end of each frame, the monitor can reallocate its resources (\textit{i.e.}, determine which equivalence relations are in place). 

For each frame, the monitor can use the approach proposed in Section 4.1, allowing it to adjust the allocation of its resources. However, this alone is not sufficient to make the monitor rational. The presence of time windows does not inherently ensure dynamicity in the monitor's visibility. The parameters that need to change include the formula under examination and the associated payoff. Without a new formula, the monitor would break the same equivalence relations in each time frame.

To update the formula, we propose removing the sub-formulas that have been verified (or violated) during the previous frames. With this new formula, we can generate a new payoff. Using this updated payoff, the monitor can adapt its strategy via the optimisation algorithm, selecting new equivalence relations to break.


Algorithm~\ref{alg:reactive_monitor} outlines the steps required to synthesize and apply a reactive monitor.
Notably, the structure of Algorithm~\ref{alg:reactive_monitor} closely resembles that of Algorithm~\ref{alg:active_monitor}, particularly in the initialisation steps (lines 1--4), where the Knapsack algorithm updates the indistinguishability relation and synthesises the monitor for the formula $\varphi$. The primary difference lies in the introduction of an if statement (lines 6--10), where the indistinguishability relation ($\sim$) and the trace ($\sigma$) are updated. This update occurs whenever the time window (specified as input to the algorithm) expires (checked on line 6). At this point, the reactive monitor updates the formula $\varphi$ based on past events observed in $\sigma$ (line 7), reflecting what the monitor has verified considering these past events. Subsequently, the payoff is updated in accordance with the new formula $\varphi$ (line 8). The Knapsack algorithm is then iteratively applied to update the indistinguishability relation (lines 9--10).
The rest of the algorithm proceeds in the same manner as Algorithm~\ref{alg:active_monitor}.

{
\begin{algorithm}[t]
\caption{ReactiveMonitor $\langle \sigma,\varphi,\sim,\pay, cost, \mathbf{b}, window \rangle$}
\label{alg:reactive_monitor}
\begin{algorithmic}[1]
\State{$count = 1$}
\State{$break = knapsack(\pay, cost, \mathbf{b})$}
\State{$\sim' \;= \;(\sim \setminus\; break)$}
\State{$Mon_\varphi = Monitor(\varphi)$}
\While{$count \leq |\sigma|$}
\If{$count \textbf{ mod } window = 0$}
\State{update $\varphi$ according to $Mon_\varphi(\sigma[1,count])$}
\State{update $\pay$ according to $\varphi$}
\State{$break = knapsack(\pay, cost, \mathbf{b})$}
\State{$\sim' \;= \;(\sim \setminus\; break)$}
\EndIf
\State{$\sigma[count] = \sigma[count] \setminus \sim'$}
\State{$count = count + 1$}
\EndWhile
\State{\textbf{return } $Mon_\varphi(\sigma[1,count])$}
\end{algorithmic}
\end{algorithm}
}

\begin{remark}
Algorithm~\ref{alg:reactive_monitor} updates the LTL formula at execution time solely to apply the metric and calculate the new payoffs for breaking equivalence classes, which is necessary due to the metric being defined on the LTL formula's structure. However, the monitor itself does not require updates during the execution of the system. Thus, the resulting Moore machine used for evaluating $\trace$ is independent of the current state of $\sim$, as demonstrated in Figure~\ref{fsm-steps-fig-ex} (\textit{i.e.}, $\sim$ is not an input for the monitor synthesis). The updated $\sim$ is only relevant for identifying the appropriate events based on the monitor's current visibility (line 11).
\end{remark}

\begin{example}\label{ex:reactive}
    Considering one last time the running example presented in Section~\ref{sec:case-study} and Examples~\ref{ex:imp_1}--\ref{ex:active}.
    Let us assume the monitor needs to verify a combination of the formulas previously introduced. For example, we may consider a new formula $\psi=\psi_1 \lor \psi_2$.
    Now, if we consider an ActiveMonitor, with the same resource bound as in Example~\ref{ex:active}, we can easily note that the monitor cannot determine a different truth value of the monitor introduced in Section~\ref{sec:rv-imperfect-info}. The reason lies in the need of breaking two equivalence classes; however, since the cost of breaking both classes (\textit{i.e.}, $5$) is greater than the bound (\textit{i.e.}, $3$), the ActiveMonitor has not the ability to conclude.
    This brings us to the use of a ReactiveMonitor instead, which, differently from the ActiveMonitor counterpart, is capable of selecting which equivalence classes to break dynamically. 
    Suppose that we have a time window of $2$, the ReactiveMonitor can break first $\equivclass_{\alpha\beta\gamma}$ and by doing so falsifying $\psi_2$ in $\trace_v(1)$. Then, the monitor can break $\equivclass_{cs}$ in the second time window (\textit{i.e.}, $\trace_v(2)$ and $\trace_v(3)$). By doing so, it can falsify $\psi_1$ in $\trace_v(3)$.
    However, notice that this resolution depends on the chosen payoff. In fact, if the monitor had selected first $\equivclass_{cs}$ and then $\equivclass_{\alpha\beta\gamma}$ it would not have concluded the violation of $\phi$.
    In our example, we have the following metrics: 
    $$metric_{\psi}(c) = (metric_{\psi_1} + metric_{\psi_2}) / 2 = $$
    $$(((0.3 \times 0) + (0.7 \times ((0.0)+(0.5 \times 1))/2)) + 0.0) / 2 = (0.7 \times 0.25) = 0.175$$
    $$metric_{\psi}(s) = 0.0$$
    $$metric_{\psi}(\gamma) = (metric_{\psi_1} + metric_{\psi_2}) / 2 = $$
    $$(0.0 + ((0.3 \times 0)+(0.7 \times ((1.0 + 0.0)/2)))) / 2 = (0.7 \times 0.5) / 2  = 0.175$$
    $$metric_{\psi}(\alpha) = 0.0$$
    $$metric_{\psi}(\beta) = 0.0$$
    Since we assumed that when two equivalence classes have the same payoff the one having the greater number of atoms is chosen, in our example the ReactiveMonitor would break $\equivclass_{\alpha\beta\gamma}$ first.
\end{example}

To conclude this section, Figure~\ref{fig:overview} recaps an overview of our approach's pipeline. Beginning with an LTL formula $\varphi$, a payoff function is generated (step 1) and used to apply the Knapsack algorithm (step 2). Following this, the monitor is synthesised (step 3). If the entire trace $\sigma$ has been analysed, the monitor can return the outcome (step 4b). If not (step 4a), the formula is revised in accordance with $\sigma$ (step 5), and a new payoff is generated (step 6), initiating another iteration of the pipeline. Note that when considering an active monitor, steps 4a, 5, and 6 can be omitted since we are dealing with a single time window.

\tikzset{every picture/.style={line width=0.75pt}} 

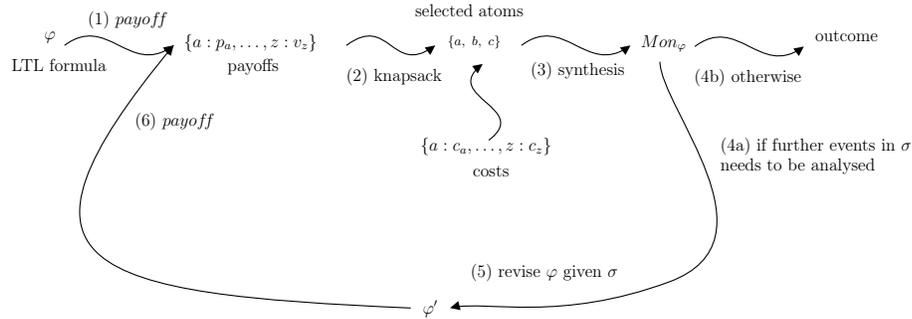
\begin{figure}

\scalebox{0.55}{

\begin{tikzpicture}[x=0.75pt,y=0.75pt,yscale=-1,xscale=1]

\draw    (58,48) .. controls (97.2,18.6) and (117.19,75.64) .. (155.62,49.69) ;
\draw [shift={(158,48)}, rotate = 143.13] [fill={rgb, 255:red, 0; green, 0; blue, 0 }  ][line width=0.08]  [draw opacity=0] (8.93,-4.29) -- (0,0) -- (8.93,4.29) -- cycle    ;
\draw    (316.5,48) .. controls (355.7,18.6) and (357.92,75.64) .. (395.65,49.69) ;
\draw [shift={(398,48)}, rotate = 143.13] [fill={rgb, 255:red, 0; green, 0; blue, 0 }  ][line width=0.08]  [draw opacity=0] (8.93,-4.29) -- (0,0) -- (8.93,4.29) -- cycle    ;
\draw    (447,135) .. controls (486.2,105.6) and (404.87,95.41) .. (439.25,66.78) ;
\draw [shift={(441.5,65)}, rotate = 143.13] [fill={rgb, 255:red, 0; green, 0; blue, 0 }  ][line width=0.08]  [draw opacity=0] (8.93,-4.29) -- (0,0) -- (8.93,4.29) -- cycle    ;
\draw    (477,48) .. controls (516.2,18.6) and (536.19,75.64) .. (574.62,49.69) ;
\draw [shift={(577,48)}, rotate = 143.13] [fill={rgb, 255:red, 0; green, 0; blue, 0 }  ][line width=0.08]  [draw opacity=0] (8.93,-4.29) -- (0,0) -- (8.93,4.29) -- cycle    ;
\draw    (636,49) .. controls (675.2,19.6) and (695.19,76.64) .. (733.62,50.69) ;
\draw [shift={(736,49)}, rotate = 143.13] [fill={rgb, 255:red, 0; green, 0; blue, 0 }  ][line width=0.08]  [draw opacity=0] (8.93,-4.29) -- (0,0) -- (8.93,4.29) -- cycle    ;
\draw    (604,63) .. controls (604.5,105) and (715.5,227) .. (612.5,265) .. controls (512.59,301.86) and (446.55,282.27) .. (414.39,288.36) ;
\draw [shift={(411.5,289)}, rotate = 345.53] [fill={rgb, 255:red, 0; green, 0; blue, 0 }  ][line width=0.08]  [draw opacity=0] (8.93,-4.29) -- (0,0) -- (8.93,4.29) -- cycle    ;
\draw    (375.5,289) .. controls (75.5,276) and (-12.5,253) .. (158,48) ;
\draw [shift={(158,48)}, rotate = 129.75] [fill={rgb, 255:red, 0; green, 0; blue, 0 }  ][line width=0.08]  [draw opacity=0] (8.93,-4.29) -- (0,0) -- (8.93,4.29) -- cycle    ;

\draw (38,34) node [anchor=north west][inner sep=0.75pt]  [font=\large] [align=left] {$\displaystyle \varphi $};
\draw (78,15) node [anchor=north west][inner sep=0.75pt]  [font=\large] [align=left] {(1) $\pay$};
\draw (166,36) node [anchor=north west][inner sep=0.75pt]  [font=\large] [align=left] {$\displaystyle \{a:p_{a} ,\dotsc ,z:v_{z}\}$};
\draw (206,59) node [anchor=north west][inner sep=0.75pt]  [font=\large] [align=left] {payoffs};
\draw (8,59) node [anchor=north west][inner sep=0.75pt]  [font=\large] [align=left] {LTL formula};
\draw (383,131) node [anchor=north west][inner sep=0.75pt]  [font=\large] [align=left] {$\displaystyle \{a:c_{a} ,\dotsc ,z:c_{z}\}$};
\draw (431,157) node [anchor=north west][inner sep=0.75pt]  [font=\large] [align=left] {costs};
\draw (408,37) node [anchor=north west][inner sep=0.75pt]   [align=left] {$\displaystyle \{a,\ b,\ c\}$};
\draw (378,9) node [anchor=north west][inner sep=0.75pt]  [font=\large] [align=left] {selected atoms};
\draw (315,67) node [anchor=north west][inner sep=0.75pt]  [font=\large] [align=left] {(2) knapsack};
\draw (484,62) node [anchor=north west][inner sep=0.75pt]  [font=\large] [align=left] {(3) synthesis};
\draw (585,37) node [anchor=north west][inner sep=0.75pt]  [font=\large] [align=left] {$\displaystyle Mon_{\varphi }$};
\draw (745,33) node [anchor=north west][inner sep=0.75pt]  [font=\large] [align=left] {outcome};
\draw (658,131) node [anchor=north west][inner sep=0.75pt]  [font=\large] [align=left] {(4a) if further events in $\displaystyle \sigma $ \\needs to be analysed};
\draw (385,280) node [anchor=north west][inner sep=0.75pt]  [font=\large] [align=left] {$\displaystyle \varphi '$};
\draw (429,247) node [anchor=north west][inner sep=0.75pt]  [font=\large] [align=left] {(5) revise $\displaystyle \varphi $ given $\displaystyle \sigma $};
\draw (121,108) node [anchor=north west][inner sep=0.75pt]  [font=\large] [align=left] {(6) $\pay$};
\draw (634,68) node [anchor=north west][inner sep=0.75pt]  [font=\large] [align=left] {(4b) otherwise};

\end{tikzpicture}

}
    \caption{Overview of the pipeline of the work.}
    \label{fig:overview}
\end{figure}


\section{Implementation}
\label{sec:implementation}

The prototype implementing the theory discussed in this paper is publicly available on GitHub\footnote{\scriptsize\url{https://github.com/AngeloFerrando/RationalMonitor}}. It consists of a Python script that implements the entire pipelines illustrated in Figure~\ref{fsm-steps-fig-ex} and Figure~\ref{fig:overview}. Python was chosen for its rich library ecosystem, particularly the Spot library\footnote{\scriptsize\url{https://spot.lrde.epita.fr/}}~\cite{DBLP:conf/mascots/Duret-LutzP04}, which is well-suited for automaton manipulation. Specifically for the monitor synthesis, we utilised Spot to automatically generate a Non-deterministic B\"{u}chi Automaton (NBA) from an LTL formula, corresponding to step (iii) in Figure~\ref{fsm-steps-fig-ex}, which is the most complex and computationally demanding step in the pipeline. The remainder of the pipeline was directly implemented in Python.

The prototype is encapsulated in a Python class named `RationalMonitor'. To create an instance of `RationalMonitor', the constructor requires the following inputs:

\begin{enumerate}[label=(\roman*)]
    \item An LTL formula to verify;
    \item A set of atomic propositions;
    \item One or more equivalence classes over the same set of atomic propositions;
    \item The metric function to evaluate and assign payoffs to equivalence classes;
    \item The costs associated with breaking these equivalence classes;
    \item The resource bound for the monitor;
    \item A time window;
    \item A trace of events to analyse.
\end{enumerate}

These parameters are related to the inputs required by Algorithm~\ref{alg:reactive_monitor}, which subsumes those of Algorithm~\ref{alg:active_monitor}. In fact, in case of an active monitor, the time window is not necessary and can be left unspecified.

Using this information, a finite state machine (FSM) representing the monitor, as defined in Definition~\ref{def:monitor-imperfect}, is constructed. The monitor then analyses the input trace and returns the corresponding verdict to the user. The trace is typically stored in a file (e.g., a log file). These input parameters can be supplied as command-line arguments to the tool. However, since the monitor is represented as a single data structure, it is also possible—and often practical—to import the script and use the monitor programmatically (\textit{i.e.}, the prototype can be used as a RV library as well). This flexibility is particularly useful for online verification scenarios, in addition to offline verification.

\begin{remark}
Although Algorithm~\ref{alg:active_monitor} and Algorithm~\ref{alg:reactive_monitor} are presented in an offline scenario, where the complete event trace $\trace$ is provided as input (e.g., from a log file), our implementation adopts an incremental approach. In practice, the monitor operates on a growing prefix of the trace, reflecting the ongoing execution of the system. This incremental adaptation does not alter the core logic of our algorithms; it is a minimal change made to align with practical use. We presented the offline version in Section~\ref{sec:rv-rational} to enhance readability and clarify the rationale behind selecting which indistinguishability relations to break.
\end{remark}

Having discussed the theoretical foundations of our approach and introduced the resulting prototype, we now turn our attention to the experiments conducted using this prototype. These experiments include applications to the remote inspection case study presented in this paper, as well as stress test scenarios designed to explore the prototype’s performance and its interaction with various metrics.

\subsection{Experimental results}

We evaluated our implementation from three distinct perspectives. 

First, we assessed our tool based on two critical aspects: generation time and verification time. The generation time refers to the execution time required to synthesise a monitor given an LTL formula (as defined in Definition~\ref{def:monitor-imperfect}). On the other hand, the verification time concerns the execution time needed to analyse a trace of events using the synthesised monitor. It is important to distinguish between these two evaluations, as monitor generation typically occurs offline, prior to system execution. Therefore, the most crucial factor in assessing runtime verification techniques is the verification time, as it directly impacts system performance by introducing runtime overhead.

Next, we validated our technique against the remote inspection case study presented in Section~\ref{sec:case-study}, ensuring that our expected outcomes aligned with the event traces analysed throughout the study.

Finally, we tested our technique by varying the metrics used to assign payoffs to the indistinguishability relations. These experiments served two purposes: to demonstrate the impact of selecting effective metrics versus poor ones, and to stress-test our implementation using a set of randomly generated LTL properties and traces.

\subsubsection{Overhead experiments.}

We began with the overhead experiments, conducting tests for both key aspects: monitor synthesis and verification. For the synthesis experiments, we varied the size of the LTL formula, defined by the number of operators within it. The size was chosen as the target for these experiments because it directly drives the monitor's generation\footnote{It is important to note that steps (iii) and (vi) in Figure~\ref{fsm-steps-fig-ex} are particularly computationally expensive, requiring exponential time relative to the formula size.}. 

For the verification experiments, we varied the length of the trace of events to be analysed. The trace length was selected because it is the only factor that influences the monitor's verification time. This can be easily understood by considering that once the FSM is generated, it remains unchanged. Its size is fixed, determined by the formula's complexity. Therefore, at runtime, the only variable affecting verification time is the length of the trace, which consists of events generated during system execution.

\begin{figure}
\centering
\begin{subfigure}{.49\textwidth}
  \centering
  \includegraphics[width=\linewidth]{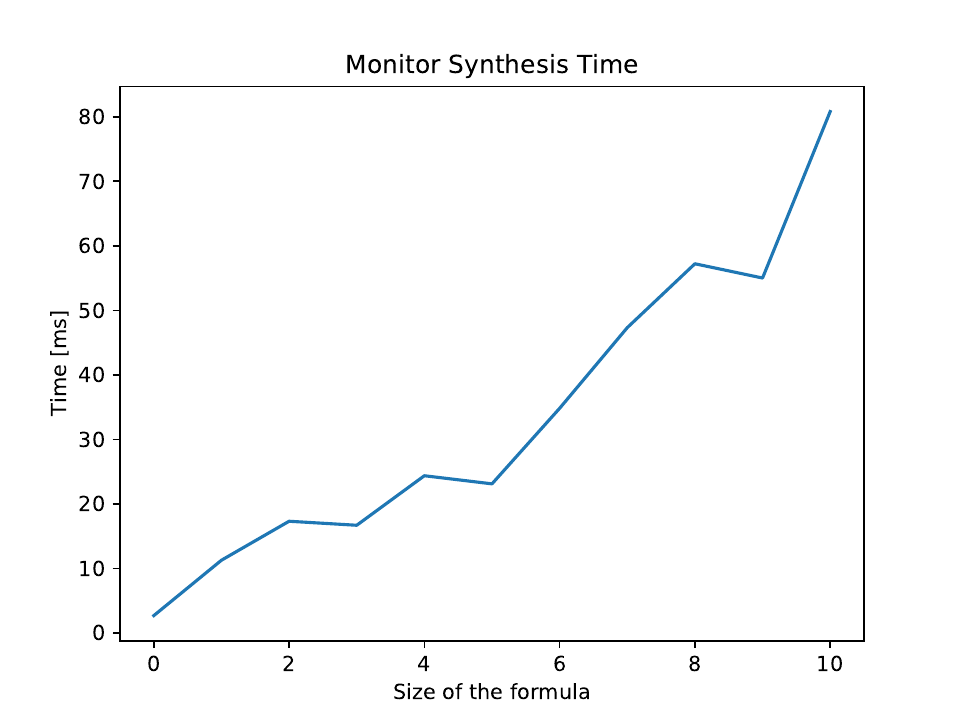}
  \caption{Time to synthesise a monitor.}
  \label{fig:generation}
\end{subfigure}
\hfill
\begin{subfigure}{.49\textwidth}
  \centering
  \includegraphics[width=\linewidth]{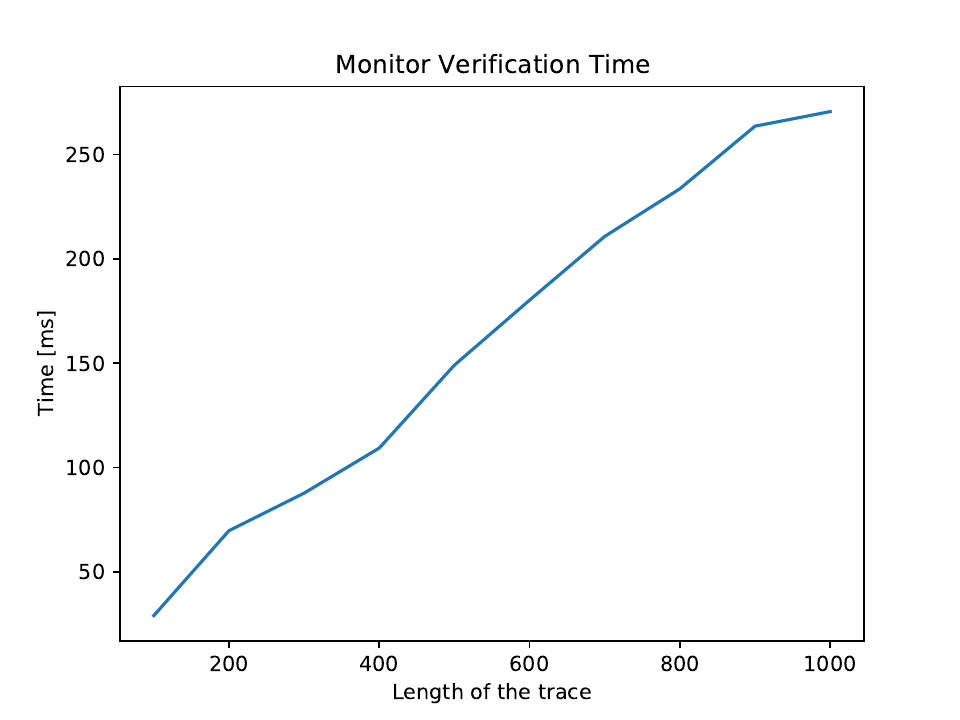}
  \caption{Time to verify a trace.}
  \label{fig:verification}
\end{subfigure}
\caption{Experimental results.}
\label{fig:experiments}
\end{figure}

Figure~\ref{fig:experiments} presents the results of our experiments, where both LTL formulas and traces were randomly generated. Specifically, Figure~\ref{fig:generation} shows the execution time required to synthesise a monitor given an LTL formula, while Figure~\ref{fig:verification} displays the execution time needed to analyse a trace of events using the synthesised monitor. 
In Figure~\ref{fig:generation}, the x-axis represents the size of the LTL formula, and the y-axis shows the execution time in milliseconds. As anticipated, the execution time for monitor synthesis increases exponentially with the size of the formula. 
In Figure~\ref{fig:verification}, the x-axis corresponds to the length of the event trace, and the y-axis again represents execution time in milliseconds. Notably, the execution time grows linearly with the trace length, which is critical for using the monitor at runtime while the system is operational. Since the execution time is linear relative to the trace length, the time required for the monitor to analyse each individual event in the trace remains constant. This ensures that the monitor can incrementally analyse events as they are generated by the system during runtime\footnote{By ``incrementally'', we mean that the monitor analyses events one by one, as opposed to offline runtime verification where the monitor expects the entire trace at once.}. 

\begin{remark}
It is important to note that the synthesis and execution process is identical for all imperfect information monitors, including the extended LTL monitor, the active monitor, and the reactive monitor. In essence, the active and reactive monitors are equivalent to the imperfect information monitor (as defined in Definition~\ref{def:monitor-imperfect}). In fact, they merely permit the breaking of indistinguishability relations. However, this modification does not impact the performance, synthesis, and execution of the underlying monitor.
\end{remark}

\subsubsection{Experiments on the case study.}

The second set of experiments we conducted focused on the case study and its formal properties. Specifically, we tested our implementation against each property presented in Section~\ref{sec:case-study} using all the types of monitors discussed in this paper. The results of these experiments are summarised in Table~\ref{tab:3v}. For each property, we considered the global trace of events described in Example~\ref{ex:imp_1}. Depending on the monitor type used, we reported the corresponding visible trace as observed by that monitor.

For instance, the standard monitor would observe the imperfect information trace without recognising that the missing events are not false but simply absent. In contrast, imperfect information monitors would consider both explicit events (as defined in Definition~\ref{def:track}) and the equivalence classes resulting from the indistinguishability relation. Additionally, for active and reactive monitors, we reported the visible trace that results from breaking indistinguishability relations according to the metric outlined in Example~\ref{ex:metric}.

{
\begin{table*}[t]
\footnotesize
\begin{tabular}{|c|c|c|c|c|c|c|c|c|}
\hline
\textbf{Global trace} & \textbf{Visible trace} & $\varphi_1$ & $\varphi_2$ & $\varphi_3$ & $\psi_1$ & $\psi_2$ & $\psi_3$ & $\psi_1 \lor \psi_2$ \\ \hline
\multicolumn{9}{c}{LTL monitor} \\ \hline
$\{\}$ & $\{\}$ & & & & & & & \\
$\{\gamma,b^1,c\}$ & $\{b^1\}$ & & & & & & & \\
$\{\gamma,c,mb,b^2\}$ & $\{mb,b^2\}$ & $?$ & $?$ & $\top$ & $?$ & $?$ & $\bot$ & $?$ \\
$\{c\}$ & $\{\}$ & & & & & & & \\
$\{w\}$ & $\{w\}$ & & & & & & & \\ \hline
\multicolumn{9}{c}{LTL monitors with imperfect information} \\ \hline
$\{\}$ & $\{ b^1_\bot,b^2_\bot,b^3_\bot, [\gamma_{cs}]_\bot, [\gamma_{\alpha\beta\gamma}]_\bot, mb_\bot, w_\bot \}$ & & & & & & & \\
$\{\gamma,b^1,c\}$ & $\{ b^1_\top,b^2_\bot,b^3_\bot, mb_\bot, w_\bot \}$ & & & & & & & \\
$\{\gamma,c,mb,b^2\}$ & $\{ b^1_\bot,b^2_\top,b^3_\bot, mb_\top, w_\bot \}$ & $?_{\not\bot}$ & $?_{\not\bot}$ & $?_{\not\bot}$ & $?_{\not\top}$ & $?_{\not\top}$ & $?_{\not\top}$ & $?_{\not\top}$ \\
$\{c\}$ & $\{ b^1_\bot,b^2_\bot,b^3_\bot, [\gamma_{\alpha\beta\gamma}]_\bot, mb_\bot, w_\bot \}$ & & & & & & & \\
$\{w\}$ & $\{ b^1_\bot,b^2_\bot,b^3_\bot, [\gamma_{cs}]_\bot, [\gamma_{\alpha\beta\gamma}]_\bot, mb_\bot, w_\top \}$ & & & & & & & \\ \hline
\multicolumn{9}{c}{Active monitor 1} \\ \hline
$\{\}$ & $\{ b^1_\bot,b^2_\bot,b^3_\bot, c_\bot,s_\bot, [\gamma_{\alpha\beta\gamma}]_\bot, mb_\bot, w_\bot \}$ & & & & & & & \\
$\{\gamma,b^1,c\}$ & $\{ b^1_\top,b^2_\bot,b^3_\bot,c_\top,s_\bot, mb_\bot, w_\bot \}$ & & & & & & & \\
$\{\gamma,c,mb,b^2\}$ & $\{ b^1_\bot,b^2_\top,b^3_\bot,c_\top,s_\bot, mb_\top, w_\bot \}$ & $\top$ & $?_{\not\bot}$ & $?$ & $\bot$ & $?_{\not\top}$ & $?_{\not\top}$ & $?_{\not\top}$ \\
$\{c\}$ & $\{ b^1_\bot,b^2_\bot,b^3_\bot,c_\top,s_\bot,[\gamma_{\alpha\beta\gamma}]_\bot, mb_\bot, w_\bot  \}$ & & & & & & & \\
$\{w\}$ & $\{ b^1_\bot,b^2_\bot,b^3_\bot, c_\bot,s_\bot, [\gamma_{\alpha\beta\gamma}]_\bot, mb_\bot, w_\top  \}$ & & & & & & & \\ 
\hline
\multicolumn{9}{c}{Active monitor 2} \\ \hline
$\{\}$ & $\{ b^1_\bot,b^2_\bot,b^3_\bot, [\gamma_{cs}]_\bot, \alpha_\bot, \beta_\bot, \gamma_\bot, mb_\bot, w_\bot \}$ & & & & & & & \\
$\{\gamma,b^1,c\}$ & $\{ b^1_\top,b^2_\bot,b^3_\bot, \alpha_\bot, \beta_\bot, \gamma_\top, mb_\bot, w_\bot \}$ & & & & & & & \\
$\{\gamma,c,mb,b^2\}$ & $\{ b^1_\bot,b^2_\top,b^3_\bot, \alpha_\bot, \beta_\bot, \gamma_\top, mb_\top, w_\bot \}$ & $\top$ & $?_{\not\bot}$ & $?$ & $?_{\not\top}$ & $\bot$ & $?_{\not\top}$ & $?_{\not\top}$ \\
$\{c\}$ & $\{ b^1_\bot,b^2_\bot,b^3_\bot,\alpha_\bot, \beta_\bot, \gamma_\bot, mb_\bot, w_\bot  \}$ & & & & & & & \\
$\{w\}$ & $\{ b^1_\bot,b^2_\bot,b^3_\bot, [\gamma_{cs}]_\bot, \alpha_\bot, \beta_\bot, \gamma_\bot, mb_\bot, w_\top  \}$ & & & & & & & \\ 
\hline
\multicolumn{9}{c}{Reactive monitor} \\ \hline
$\{\}$ & $\{ b^1_\bot,b^2_\bot,b^3_\bot, [\gamma_{cs}]_\bot, \alpha_\bot, \beta_\bot, \gamma_\bot, mb_\bot, w_\bot \}$ & & & & & & & \\
$\{\gamma,b^1,c\}$ & $\{ b^1_\top,b^2_\bot,b^3_\bot, \alpha_\bot, \beta_\bot, \gamma_\top, mb_\bot, w_\bot \}$ & & & & & & & \\
$\{\gamma,c,mb,b^2\}$ & $\{ b^1_\bot,b^2_\top,b^3_\bot, c_\top, s_\bot, mb_\top, w_\bot \}$ & $\top$ & $?_{\not\bot}$ & $?$ & $\bot$ & $\bot$ & $?_{\not\top}$ & $\bot$ \\
$\{c\}$ & $\{ b^1_\bot,b^2_\bot,b^3_\bot, c_\top, s_\bot, [\gamma_{\alpha\beta\gamma}]_\bot, mb_\bot, w_\bot \}$ & & & & & & & \\
$\{w\}$ & $\{ b^1_\bot,b^2_\bot,b^3_\bot, [\gamma_{cs}]_\bot, [\gamma_{\alpha\beta\gamma}]_\bot, mb_\bot, w_\top \}$ & & & & & & & \\ \hline
\end{tabular}
\caption{Results on applying the prototype implementation on the case study. Note that, Active monitor 1 breaks $\gamma_{cs}$ while Active monitor 2 breaks $\gamma_{\alpha\beta\gamma}$.}
\label{tab:3v}
\end{table*}
}

Finally, in Table~\ref{tab:3v}, we documented the outcomes produced by each monitor for each property, considering the given trace, thus validating the results discussed throughout the paper. Notably, we observed that the use of imperfect information monitors corrected the erroneous outcomes for properties $\varphi_3$ and $\psi_3$, which the standard monitor incorrectly classified as $\top$ and $\bot$, respectively. 

Moreover, Table~\ref{tab:3v} illustrates the effect of using a reactive monitor in the verification of $\psi = \psi_1 \lor \psi_2$. The imperfect information monitor could at best return $?_{\not\top}$, similar to the active monitor. The active monitor did not improve the outcome for $\psi$ because the bound $\textbf{b}$ was set to 3, forcing the monitor to choose between breaking one equivalence class (e.g., $\gamma_{cs}$) or the other (e.g., $\gamma_{\alpha\beta\gamma}$), but not both.
As a result, the reactive monitor achieved better outcomes than its counterparts by dynamically adapting to and reacting to new events. By doing so, and with a time window set to 2 (as demonstrated in Example~\ref{ex:reactive}), the reactive monitor was able to first break one equivalence class ($\gamma_{\alpha\beta\gamma}$) and then the other ($\gamma_{cs}$) in two separate steps during the trace of events analysed.

\subsubsection{Experiments on metrics.}

The final set of experiments focused on the impact of different metrics on the results obtained by the monitors. Since these experiments were solely concerned with the influence of metrics, we opted to use active monitors instead of reactive ones. This choice was made because the objective was to isolate the effect of metrics on the monitoring process, whereas the impact of reactive behaviour was already addressed in the previous set of experiments.

Specifically, we tested four different metrics, as detailed in Table~\ref{table:metrics_comparison} and Table~\ref{table:metrics_comparison_perc}. We generated 10,000 random LTL formulas and verified them against 100 randomly generated traces of events, resulting in 1,000,000 executions for each metric analysed. It is important to note that while the formulas and traces were randomly generated, they were kept consistent across the different experiments. This ensured that each metric was tested on the same set of LTL formulas and traces, allowing for a fair comparison.

\begin{table}[h!]
\centering
\begin{tabular}{|c|c|c|c|c|c|c|c|}
\hline
\textbf{Metric} & \textbf{Total} & \textbf{$\top$} & \textbf{$\bot$} & \textbf{$uu$} & \textbf{$\unknown$} & \textbf{$?_{\not\bot}$} & \textbf{$?_{\not\top}$} \\ \hline
\textbf{$metric_0$} & 1000000 & 175793 & 275205 & 29676 & 463563 & 29873 & 25890 \\ \hline
\textbf{$metric_1$} & 1000000 & 165239 & 280671 & 16983 & 481076 & 26978 & 29053 \\ \hline
\textbf{$metric_2$} & 1000000 & 173899 & 278783 & \textbf{16875} & 480616 & 29066 & 20761 \\ \hline
\textbf{$metric_3$} & 1000000 & 163404 & 284628 & 17498 & 478891 & 24701 & 30878 \\ \hline
\end{tabular}
\caption{Comparison of Metrics.}
\label{table:metrics_comparison}
\end{table}

Table~\ref{table:metrics_comparison} presents the results of these experiments, showing the number of times each outcome was returned by the monitor for each metric. Table~\ref{table:metrics_comparison_perc} provides the same results in percentage form, offering a clearer understanding of the metrics' influence.

\begin{table}[h!]
\centering
\begin{tabular}{|c|c|c|c|c|c|c|c|}
\hline
\textbf{Metric} & \textbf{Total} & \textbf{$\top$} & \textbf{$\bot$} & \textbf{$uu$} & \textbf{$\unknown$} & \textbf{$?_{\not\bot}$} & \textbf{$?_{\not\top}$} \\ \hline
\textbf{$metric_0$} & 1000000 & 17.58\% & 27.52\% & 2.97\% & 46.36\% & 2.99\% & 2.59\% \\ \hline
\textbf{$metric_1$} & 1000000 & 16.52\% & 28.07\% & 1.70\% & 48.11\% & 2.70\% & 2.91\% \\ \hline
\textbf{$metric_2$} & 1000000 & 17.39\% & 27.88\% & \textbf{1.69\%} & 48.06\% & 2.91\% & 2.08\% \\ \hline
\textbf{$metric_3$} & 1000000 & 16.34\% & 28.46\% & 1.75\% & 47.89\% & 2.47\% & 3.09\% \\ \hline
\end{tabular}
\caption{Comparison of Metrics in Percentages.}
\label{table:metrics_comparison_perc}
\end{table}

We labelled the metrics from 0 to 3, referring to each as $metric_i$ for brevity. Among these, $metric_2$ corresponds to the metric introduced in Example~\ref{ex:metric}, while the others are variations that assign different weights to temporal operators (see the Appendix for the detailed description). Specifically, $metric_0$ serves as a baseline metric, as it assigns payoffs to atomic propositions without giving particular importance to any of them. We included this baseline metric to demonstrate how selecting a more carefully considered metric can significantly improve the monitoring process. This improvement is evident in Table~\ref{table:metrics_comparison}, where $metric_0$ shows the poorest performance, with the highest number of undefined outcomes. In contrast, $metric_2$ performs the best, yielding the fewest undefined outcomes. Notice that, the undefined value is the worst outcome between the six truth values. In fact, it gives a definitive result without providing any information on the analysed formula. 


%

\begin{remark}
It is important to note that while $metric_2$ performs better than $metric_0$ in terms of reducing the number of undefined outcomes, the percentage of such outcomes remains minimal compared to the other monitor results (with $?$ being the most frequently returned outcome). This indicates that although selecting a well-suited metric can indeed reduce the proportion of undefined outcomes, its impact is more pronounced in scenarios where the metric is contextually relevant. 
In the experiments summarised in Table~\ref{table:metrics_comparison} and Table~\ref{table:metrics_comparison_perc}, the random generation of LTL formulas and traces demonstrates that the choice of a metric does not significantly affect the overall distribution of outcomes. However, it does play a crucial role in minimising the occurrence of the least desirable outcome, namely the undefined result.
\end{remark}




\section{Related Work}
\label{sec:related-work}


Several works address runtime verification in scenarios involving uncertainty; for a comprehensive overview, we refer to a recent survey on the topic~\cite{DBLP:journals/csr/TalebHK23}.

The work most closely related to ours is~\cite{DBLP:conf/rv/WangASL11}, where Past-Time LTL is verified at runtime under uncertainty regarding the observed events. In that study, verification is performed on abstract traces of events, where not all concrete events are present—only samples taken at specific time intervals. The uncertainty arises from unknown event interleaving, whereas in our approach, it stems from indistinguishability relations between events. Unlike~\cite{DBLP:conf/rv/WangASL11}, we do not sample events; instead, the uncertainty is determined by the monitor's visibility. Consequently, our abstraction focuses not on the order of events in a trace but on the types of events the trace contains.

Another closely related work to ours is presented in~\cite{DBLP:conf/sac/JoshiTF17}, where the authors propose an approach to RV in scenarios involving transient event loss—where events are temporarily lost, their quantity is known, but their content remains unknown. The objective of their work is to demonstrate that certain properties can still be monitored despite these losses, provided that subsequent valid events are observed. They denote lost events using the symbol $\chi$ and express the properties using LTL. These properties are then translated into an RV-LTL monitor, which is represented as a finite-state automaton. To create a loss-tolerant monitor, they introduce a new state that manages the lossy elements. The monitor outputs a verdict of $?$ (analogous to our $uu$ verdict) upon encountering a lost event but continues its operation based on the subsequent valid events. Unlike our approach, which also considers online monitoring and the management of imperfect information through active and reactive monitors, the approach in~\cite{DBLP:conf/sac/JoshiTF17} is confined to an offline setting and does not address the handling of imperfect information.

The works in~\cite{DBLP:conf/rv/BasinKMZ12,DBLP:conf/rv/BasinKMZ14,DBLP:conf/cav/BasinKZ17} and ours both address RV under challenging conditions, yet they differ significantly in their approaches to handling imperfect information and the types of monitors employed. The methods in~\cite{DBLP:conf/rv/BasinKMZ12,DBLP:conf/rv/BasinKMZ14,DBLP:conf/cav/BasinKZ17} emphasise runtime verification techniques that manage incomplete or conflicting logs, imprecise timestamps, and out-of-order data streams, ensuring effective monitoring despite these adversities. In contrast, our work advances standard RV by explicitly addressing imperfect information through the use of indistinguishability relations, and by developing a novel ``rational monitor'' capable of dynamically adjusting its behaviour based on the information it receives. This proactive approach contrasts with the more static methods used in~\cite{DBLP:conf/rv/BasinKMZ12,DBLP:conf/rv/BasinKMZ14,DBLP:conf/cav/BasinKZ17}. Furthermore, while the methods in~\cite{DBLP:conf/rv/BasinKMZ12,DBLP:conf/rv/BasinKMZ14,DBLP:conf/cav/BasinKZ17} are typically applied within traditional RV frameworks, often in compliance monitoring or real-time systems, our approach is specifically designed for autonomous systems. This emphasis on a monitor that remains effective even under severely limited or faulty information makes our approach particularly suitable for complex, real-world environments such as autonomous robotics.

The work in~\cite{DBLP:conf/fsttcs/AcetoAFI17} addresses the monitorability of branching-time logics that incorporate silent actions, focusing on how these actions influence the ability to verify certain properties at runtime. Their research is particularly concerned with the challenges posed by silent actions in RV and develops a framework to assess the monitorability of specifications in the presence of these actions. In contrast, our work expands on the concept of RV by addressing scenarios where the information available is imperfect, utilising indistinguishability relations to manage and verify systems under these conditions. Furthermore, we introduce the consideration of rational aspects in RV, particularly under conditions of imperfect information. While both studies contribute to the advancement of RV techniques, our approach offers a broader perspective by encompassing various types of imperfect information, whereas~\cite{DBLP:conf/fsttcs/AcetoAFI17} is more specialised, dealing with the specific implications of silent actions within a RV context.

In a different line of research, works such as~\cite{DBLP:conf/rv/StollerBSGHSZ11,DBLP:conf/rv/BartocciGKSSZS12,DBLP:conf/rv/KalajdzicBSSG13,DBLP:journals/corr/BartocciG13,DBLP:conf/rv/LeuckerSS0T19} address uncertainty in verification caused by the absence of information. In these studies, event traces may contain gaps, meaning that the monitor cannot observe system behaviour at certain points during execution. This issue has been addressed in various ways, typically by filling these gaps with potential events. Given the uncertainty about which events actually occurred during these gaps, these approaches often rely on probabilistic methods to infer the missing events. These works differ fundamentally from ours because we do not assume any missing information—there are no gaps in our traces. Instead, our uncertainty arises from the monitor's inability to distinguish between certain events due to indistinguishability relations.

A more recent work on RV with uncertainty is~\cite{DBLP:conf/icse/TalebKH21}, where uncertainty is abstracted using multi-traces instead of standard uni-traces. A multi-trace allows multiple evaluations for the same atomic proposition within the trace. The authors propose a monitor designed to handle such multi-traces and prove its soundness. Similar to~\cite{DBLP:conf/rv/StollerBSGHSZ11,DBLP:conf/rv/BartocciGKSSZS12,DBLP:conf/rv/KalajdzicBSSG13,DBLP:journals/corr/BartocciG13,DBLP:conf/rv/LeuckerSS0T19}, this work also focuses on missing events, although it considers partially missing events as well.

Unlike our approach, all the aforementioned works explicitly represent the notion of uncertainty (e.g., through gaps). When traces contain concrete events, these works typically adhere to standard semantics. Our approach, by contrast, is less intrusive, building on the existing RV pipeline for verifying LTL properties without introducing gaps. We focus on extending the standard RV technique for LTL to handle scenarios where the monitor has imperfect information about the system. From an engineering standpoint, our method enhances the standard LTL approach for use in cases of imperfect information, while the other works in the literature primarily propose entirely new techniques to manage the absence of information, often caused by noise or technical issues.

\section{Conclusions and Future Work}
\label{sec:conclusions-and-future-work}

In this paper, we introduced several extensions to the standard LTL runtime verification approach. We addressed the challenge of imperfect information at the monitor level and demonstrated how this lack of information can cause a standard LTL monitor to produce incorrect verdicts. To mitigate this issue, we proposed an extended version of the LTL monitoring process specifically designed to handle imperfect information. However, monitors with imperfect information yield fewer verdicts than those with perfect information. To address this limitation, we introduced the novel concept of rational monitors. In particular, we defined two classes of rational monitors: active and reactive, which are designed to improve the handling of imperfect information at the monitoring level. We provided a theoretical framework for understanding imperfect information through the use of equivalence classes, alongside the concepts of active and reactive monitors, and examined how imperfect information affects the verification of LTL properties. Additionally, we presented a Python prototype implementing our approach and demonstrated its application on a relevant case study, alongside additional experiments designed to stress-test the implementation.

For future work, we envision two main directions. First, we plan to further investigate the role of metrics in the imperfect information runtime verification process. Our goal is to identify a metric that offers the best possible trade-off for verifying LTL formulas that minimises the percentage of undefined outcomes while preserving other definitive truth outcomes. 
Second, we aim to extend the concept of rational monitors to a multi-agent setting. This extension could have significant implications for real-world applications, particularly in Internet of Things (IoT) systems, where individual components are often monitored independently. Due to the distributed nature of IoT systems, these components frequently lack complete knowledge of the entire system. In such scenarios, a variant of our approach could be highly effective.


%
%
%
\bibliographystyle{splncs04}
\bibliography{main}
\newpage
\section*{Appendix}

We report the $metric_0$ as used in Section~\ref{sec:implementation}
\begin{eqnarray*}
metric_q(p) &=&  0 \; with \; p \neq q\\
metric_p(p) &=&  1\\
metric_{\lnot p}(p) &=& 1 \\ 
metric_{\spec \lor \spec'}(p) &=& (metric_{\spec}(p) + metric_{\spec'}(p)) / 2 \\
metric_{\spec \land \spec'}(p) &=& min(metric_{\spec}(p), metric_{\spec'}(p))\\
metric_{\spec\Rightarrow\spec'}(p) &=& (metric_{\spec}(p) + metric_{\spec'}(p)) / 2 \\ 
metric_{\nextOp{\spec}}(p) &=& 0.1 \times metric_{\spec}(p) \\
metric_{\spec\until\spec'}(p) &=& 0.9 \times metric_{\spec}(p) + 0.1 \times metric_{\spec'}(p)\\
metric_{\spec\release\spec'}(p) &=& 0.9 \times metric_{\spec}(p) + 0.1 \times metric_{\spec'}(p)
\end{eqnarray*}

We report the $metric_1$ as used in Section~\ref{sec:implementation}
\begin{eqnarray*}
metric_q(p) &=&  0 \; with \; p \neq q\\
metric_p(p) &=&  1\\
metric_{\lnot p}(p) &=& 1 \\ 
metric_{\spec \lor \spec'}(p) &=& (metric_{\spec}(p) + metric_{\spec'}(p)) / 2 \\
metric_{\spec \land \spec'}(p) &=& max(metric_{\spec}(p), metric_{\spec'}(p))\\
metric_{\spec\Rightarrow\spec'}(p) &=& (metric_{\spec}(p) + metric_{\spec'}(p)) / 2 \\ 
metric_{\nextOp{\spec}}(p) &=& 0.5 \times metric_{\spec}(p) \\
metric_{\spec\until\spec'}(p) &=& 0.3 \times metric_{\spec}(p) + 0.7 \times metric_{\spec'}(p)\\
metric_{\spec\release\spec'}(p) &=& 0.5 \times metric_{\spec}(p) + 0.5 \times metric_{\spec'}(p)
\end{eqnarray*}

We report the $metric_3$ as used in Section~\ref{sec:implementation}
\begin{eqnarray*}
metric_q(p) &=&  0 \; with \; p \neq q\\
metric_p(p) &=&  1\\
metric_{\lnot p}(p) &=& 1 \\ 
metric_{\spec \lor \spec'}(p) &=& (metric_{\spec}(p) + metric_{\spec'}(p)) / 2 \\
metric_{\spec \land \spec'}(p) &=& max(metric_{\spec}(p), metric_{\spec'}(p))\\
metric_{\spec\Rightarrow\spec'}(p) &=& (metric_{\spec}(p) + metric_{\spec'}(p)) / 2 \\ 
metric_{\nextOp{\spec}}(p) &=& 1.0 \times metric_{\spec}(p) \\
metric_{\spec\until\spec'}(p) &=& 0.3 \times metric_{\spec}(p) + 0.7 \times metric_{\spec'}(p)\\
metric_{\spec\release\spec'}(p) &=& 0.3 \times metric_{\spec}(p) + 0.7 \times metric_{\spec'}(p)
\end{eqnarray*}

\end{document}